\documentclass[12pt]{article}

\usepackage{amsmath,amsthm,amssymb,thm-restate} 
\usepackage{tikz,xspace,nicefrac,enumitem}
\usetikzlibrary{shapes.geometric}
\usepackage[margin=1in]{geometry}

\usepackage{hyperref}
\newcommand{\arxiv}[1]{\href{http://arxiv.org/abs/#1}{\texttt{arXiv:#1}}}

\allowdisplaybreaks

\makeatletter
\g@addto@macro\normalsize{%
  \setlength\abovedisplayskip{8pt plus 3pt minus 3pt}
  \setlength\belowdisplayskip{8pt plus 3pt minus 3pt}
  \setlength\abovedisplayshortskip{6pt plus 3pt minus 2pt}
  \setlength\belowdisplayshortskip{6pt plus 3pt minus 2pt}
}
\makeatother

\date{\small Submitted: 4 November 2019; Revised: 4 December 2020}

\def\dfrac#1#2{\lower0.15ex\hbox{\large$\frac{#1}{#2}$}}

\numberwithin{equation}{section}

\def\({\bigl(}
\def\){\bigr)}

\newtheorem{thm}{Theorem}
\newtheorem{corollary}{Corollary}

\newtheorem{lemma}{Lemma}

\let\eps=\varepsilon

\newcommand{\bpw}{\operatorname{bpw}}             
\newcommand{\btw}{\operatorname{btw}}             
\newcommand{\cA}{\mathcal{A}}                     
\newcommand{\cB}{\mathcal{B}}                     
\newcommand{\cC}{\mathcal{C}}                     
\newcommand{\cF}{\mathcal{F}}                     
\newcommand{\Nb}{\mathrm{N}}                      
\newcommand{\cI}{\mathcal{I}}                     
\newcommand{\es}{\varnothing}                     
\newcommand{\IN}{\mathbb{N}}                      
\newcommand{\lmax}{\ell_{\textup{max}}}

\newcommand{\pw}{\operatorname{pw}}               
\newcommand{\tw}{\operatorname{tw}}               
\newcommand{\sd}{\oplus}                          
\newcommand{\sm}{\setminus}
\newcommand{\fast}{fast\xspace}
\newcommand{\tprob}{\boldsymbol{P}}               
\newcommand{\E}{\textbf{E}}                       
\newcommand{\var}{\textbf{var}}                   
\newcommand{\dtv}{\textrm{d}_{\textrm{TV}}}       
\let\originalleft\left
\let\originalright\right
\renewcommand{\left}{\mathopen{}\mathclose\bgroup\originalleft}
\renewcommand{\right}{\aftergroup\egroup\originalright}

  \definecolor{lightred}{rgb}{1.0 0.8 0.8}
  \definecolor{lightblue}{rgb}{0.8 0.8 1.0}
  \definecolor{darkgreen}{rgb}{0.0 0.5 0.0}
  \definecolor{lightgreen}{rgb}{0.7 1.0 0.7}

\tikzset{ b/.style = { circle 
                     , draw
                     , thick
                     , inner sep = 0pt
                     , fill = black
                     , minimum size = 3.5pt
                     }
        , w/.style = { circle 
                     , draw
                     , thick
                     , inner sep = 0pt
                     , fill = white
                     , minimum size = 3.7pt
                     }
        , r/.style = { circle 
                     , draw
                     , thick
                     , inner sep = 0pt
                     , fill = red
                     , minimum size = 4.0pt
                     }
        , i/.style = { circle 
                     , thick
                     , inner sep = 0pt
                     , minimum size = 4.0pt
                     }
        , R/.style = { rectangle, rounded corners 
                     , draw
                     , thick
                     , inner sep = 1.5pt
                     , fill = lightred
                     , minimum size = 6.0mm
                     }
        , B/.style = { diamond 
                     , draw
                     , thick
                     , inner sep = 1.5pt
                     , fill = lightblue
                     , minimum size = 5.2mm
                     }
        , U/.style = { circle 
                     , draw
                     , thick
                     , inner sep = 0pt
                     , minimum size = 5.5mm
                     }
        }

\title{Counting independent sets in graphs\\with bounded bipartite pathwidth%
\thanks{A preliminary version of this paper appeared as~\cite{DGM}.}}
\author{Martin Dyer\thanks{School of Computing, University of Leeds, Leeds LS2~9JT, UK.
Email: \texttt{m.e.dyer@leeds.ac.uk}. Supported by EPSRC research grant EP/S016562/1
```Sampling in hereditary classes''.}
\and Catherine Greenhill\thanks{School of Mathematics and Statistics, UNSW Sydney, NSW 2052,
 Australia. Email: \texttt{c.greenhill@unsw.edu.au}. Supported by Australian Research Council grant DP190100977.}
\and Haiko M\"{u}ller\thanks{School of Computing, University of Leeds, Leeds LS2~9JT, UK.
Email: \texttt{h.muller@leeds.ac.uk}. Supported by EPSRC research grant EP/S016562/1
```Sampling in hereditary classes''.}
}

\begin{document}

\maketitle

\begin{abstract}
We show that a simple Markov chain, the Glauber dynamics, can efficiently sample independent sets almost uniformly at random in polynomial time for graphs in a certain class. The class is determined by boundedness of a new graph parameter called bipartite pathwidth.  This result, which we prove for the more general hardcore distribution with fugacity~$\lambda$, can be viewed as a strong generalisation of Jerrum and Sinclair's work on approximately counting matchings, that is, independent sets in line graphs. The class of graphs with bounded bipartite pathwidth includes claw-free graphs, which generalise line graphs. We consider two further generalisations of claw-free graphs and prove that these classes have bounded bipartite pathwidth.
We also show how to extend all our results to polynomially-bounded vertex weights.\\

\indent\textbf{Running title}: Bounded bipartite pathwidth
\end{abstract}
\section{Introduction}\label{s:intro}

There is a well-known bijection between matchings of a graph $G$ and independent sets in the line graph of $G$.
We will show that we can approximate the number of independent sets
in graphs for which all bipartite induced subgraphs are well structured, in a sense
that we will define precisely. Our approach is to generalise the Markov chain
analysis of Jerrum and Sinclair~\cite{JS} for the corresponding problem of counting matchings.

The canonical path argument given by Jerrum and Sinclair in~\cite{JS} relied on the fact
that the symmetric difference of two matchings of a given graph $G$ is a bipartite subgraph of $G$ consisting of
a disjoint union of paths and even-length cycles.
We introduce a new graph parameter, which we call bipartite pathwidth,
to enable us to give the strongest generalisation of the approach of~\cite{JS},
far beyond the class of line graphs.

\subsection{Independent set problems}\label{ss:is}

For a given graph $G$, let $\cI(G)$ be the set of all independent sets in $G$.
The \emph{independence number} $\alpha(G) = \max\{|I| \, :\, I \in \cI(G)\}$ is
the size of the largest independent set in $G$.
(We will sometimes simply denote this parameter as $\alpha$, if the graph $G$
is clear from the context.)
The problem of finding $\alpha(G)$ is NP-hard in general,
even in various restricted cases, such as degree-bounded graphs.
However, polynomial time algorithms have been constructed for computing $\alpha$, and finding an independent set
$I$ such that $\alpha=|I|$, for various graph classes.
The most important case has been \emph{matchings}, which are independent sets
in the \emph{line graph} $L(G)$ of $G$.
This has been generalised
to larger classes of graphs, for example
\emph{claw-free} graphs~\cite{Minty}, which include line graphs~\cite{beineke}, and
\emph{fork-free} graphs~\cite{Alekseev}, which include claw-free graphs.

Counting independent sets in graphs, determining $|\cI(G)|$,
is known to be \#P-complete in general~\cite{PrBa83},
and in various restricted cases \cite{greenhill,Vadhan}.
Exact counting in polynomial time is known only for some restricted graph classes, e.g.~\cite{DM}.
Even approximate counting is NP-hard in general,
and is unlikely to be in polynomial time for bipartite graphs~\cite{DGGJ03}. The relevance here of the
optimisation results above is that proving NP-hardness of approximate counting is usually based on the hardness
of some optimisation problem.

However, for some classes of graphs, for example line graphs, approximate
counting is known to be possible~\cite{JS,JSV}. The most successful approach to the problem has been
the Markov chain approach, which relies on a close correspondence between approximate counting
and sampling uniformly at random~\cite{JVV}.
The Markov chain method was applied to degree-bounded graphs in~\cite{LV} and~\cite{DG}.
In his PhD thesis~\cite{JM}, Matthews used the Markov chain approach with a Markov chain for
sampling independent sets in \emph{claw-free} graphs. His chain, and its analysis,
directly generalises that of~\cite{JS}.

Several other approaches to approximate counting have been successfully applied
to the independent set problem. Weitz~\cite{weitz} used the \emph{correlation decay} approach
on degree-bounded graphs, resulting in a deterministic polynomial time approximation
algorithm (an FPTAS) for counting independent sets in graphs with degree at most 5. Sly~\cite{sly} gave a
matching NP-hardness result. The correlation decay method was also applied to matchings in~\cite{BGKNP}, and
was extended to complex values of $\lambda$ in~\cite{HSV}.
Efthymiou et al.~\cite{EHSVY} proved that the Markov chain approach can
(almost) produce the best results obtainable by other methods,
for graphs with sufficiently high maximum degree and girth.

Very recently (after this paper was submitted), Anari et al~\cite{alg}
proved that the Glauber dynamics for weighted independent sets (the hardcore model)
is rapidly mixing in the tree uniqueness region, and their analysis was simplified
and generalised by Chen et al.~\cite{clv}, improving the mixing time.

The \emph{independence polynomial} $P_G(\lambda)$ of a graph $G$ is defined in (\ref{independence}) below.
The \emph{Taylor series} approach of Barvinok~\cite{barvinok-book,barvinok} was used by Patel and Regts~\cite{PR}
to give a FPTAS for $P_G(\lambda)$ in degree-bounded claw-free graphs. The success of the method depends on the location of the roots of the independence polynomial. 
Chudnovsky and Seymour~\cite{CS} proved that all
these roots are real, and hence they are all negative. Then the algorithm of~\cite{PR} is valid
for all complex $\lambda$ which are not real and negative.
Bencs~\cite{bencs} gave a new proof of Chudnovsky--Seymour result,
which involved establishing that independence polynomials satisfy 
certain Christoffel--Darboux identities. 
It follows from his proof that there is an open
region around the positive axis which is zero-free for the independence polynomial of
bounded-degree fork-free
graphs.  Thus the Taylor expansion approach can be applied to give an FPTAS for
$P_G(\lambda)$ in bounded-degree fork-free graphs.

Bonsma et al.~\cite{BKW} gave a polynomial-time algorithm which takes a claw-free graph $G$
and two independent sets of $G$ and decides whether one of the independent sets can be transformed
into the other by a sequence of elementary moves, each of which deletes one vertex and inserts
another, producing a new independent set. This implies ergodicity of the Glauber dynamics, though
we make no use of this result.

In this paper, we return to the Markov chain approach, providing a broad generalisation
of the methods of~\cite{JS}. In Section~\ref{s:bippathwidth} we define a graph parameter
which we call \emph{bipartite pathwidth}, and the class $\cC_p$ of graphs with
bipartite pathwidth at most $p$.  The Markov chain which we analyse is the well-known
\emph{Glauber dynamics}. We now state our main result, which gives a bound on the mixing time of the
Glauber dynamics for graphs of bounded bipartite pathwidth.
Some remarks about values of $\lambda$ not covered by Theorem~\ref{thm:main-Cp}
can be found at the end of Section~\ref{s:intro}.

\begin{restatable}{thm}{mainthm}\label{thm:main-Cp}
Let $G\in\cC_p$ be a graph with $n$ vertices and let $\lambda\geq e/n$, where $p\geq 2$ is an integer.
Then the Glauber dynamics with fugacity $\lambda$ on $\mathcal{I}(G)$
(and initial state $\es$) has mixing time
\[ \tau_\es(\eps) \leq 2e\alpha(G)\, n^{p+1}\, \lambda^p\Big(1+\max(\lambda,1/\lambda)\Big)\Big(\alpha(G)\,\ln(n\lambda)+1+\ln(1/\eps)\Big).\]
When $p$ is constant, this upper bound is polynomial in $n$ and $\max(\lambda,1/\lambda)$.
\end{restatable}

The plan of the paper is as follows. In Section~\ref{s:markov}, we define the necessary Markov chain
background and define the Glauber dynamics.  In Section~\ref{s:bippathwidth} we develop the concept of
bipartite pathwidth, and use it in Section~\ref{s:ispaths} to determine canonical paths for independent sets.
In Sections~\ref{s:claw-free} and~\ref{s:classes} we introduce some graph classes which have bounded bipartite pathwidth.  These classes, like the class of claw-free graphs, are defined by excluded induced subgraphs.

\subsection{Preliminaries}\label{ss:notation}

Let $\IN$ be the set of positive integers. We write $[m] = \{1,2,\ldots, m\}$ for any $m\in \IN$,
For any integers $a\geq b\geq 0$, write $(a)_b = a(a-1)\cdots (a-b+1)$
for the \emph{falling factorial}. Let $A\sd B$ denote the symmetric difference of sets $A,B$.

For graph theoretic definitions not given here, see~\cite{survey,diestel}.
{Throughout this paper, all graphs are simple and undirected.}
The term ``induced subgraph'' will mean a vertex-induced subgraph,
and the subgraph of $G=(V,E)$ induced by the set $S$ will be denoted by $G[S]$.
The \emph{neighbourhood} $\{u\in V: uv\in E\}$ of $v\in V$ will be denoted by $\Nb(v)$,
and we will denote $\Nb(v)\cup\{v\}$ by $\Nb[v]$.
If two graphs $G,H$ are isomorphic, we will write $G\cong H$. The \emph{complement} of a graph $G=(V,E)$ is the graph $\overline{G}=(V,\overline{E})$, where $\overline{E}=\{uv: u,v\in V,\ uv\notin E\}$.

Given a graph $G=(V,E)$, 
let $\cI_k(G)$ be the set of independent sets of $G$ of size $k$.
%
The \emph{independence polynomial} of $G$ is the \emph{partition function}

\begin{equation}
\label{independence}
P_G(\lambda)=\sum_{I\in\cI(G)}\lambda^{|I|}=\sum_{k=0}^{\alpha(G)} N_k\, \lambda^k,
\end{equation}
where $N_k=|\cI_k(G)|$ for $k=0,\ldots, \alpha$.
Here $\lambda\in \mathbb{C}$ is called the  \emph{fugacity}.
In this paper, we consider only nonnegative real $\lambda$,
We have $N_0=1$, $N_1=n$ and $N_k\leq\binom{n}{k}$ for $k=2,\ldots, n$. Thus
it follows that for any $\lambda\geq 0$,
\begin{equation}\label{eq:Pbounds}
1+n\lambda\leq P_G(\lambda)\leq \sum_{k=0}^{\alpha(G)} \binom{n}{k}\, \lambda^k\leq (1+\lambda)^n.
\end{equation}
Note also that $P_G(0)=1$ and $P_G(1)=|\cI(G)|$.

An \emph{almost uniform sampler} for a probability distribution $\pi$ on a state $\Omega$ is a
randomised algorithm which takes as input a real number $\delta > 0$ and
outputs a sample from a distribution $\mu$ such that
the \emph{total variation distance} $\frac12\sum_{x\in\Omega}|\mu(x)-\pi(x)|$ is at most $\delta$.
The sampler is a \emph{fully polynomial almost uniform sampler (FPAUS)} if its running
time is polynomial in the input size $n$ and $\log(1/\delta)$. The word ``uniform'' here is historical,
as it was first used in the case where $\pi$ is the uniform distribution. We use it in a more general setting.

If $w:\Omega\to\mathbb{R}$ is a \emph{weight function}, then the \emph{Gibbs distribution} $\pi$ satisfies
$\pi(x)=w(x)/W$ for all $x\in\Omega$,
where $W=\sum_{x\in\Omega}w(x)$. If $w(x)=1$ for all $x\in\Omega$ then $\pi$ is uniform.
For independent sets with $w(I)=\lambda^{|I|}$, the Gibbs distribution satisfies
\begin{equation}
\label{Gibbs}
 \pi(I) = \lambda^{|I|}/P_G(\lambda),
\end{equation}
and is often called the hardcore distribution. Jerrum, Valiant and Vazirani~\cite{JVV}
showed that approximating $W$ is equivalent to the existence of an FPAUS for $\pi$, provided the problem is
 \emph{self-reducible}. Counting independent sets in a graph is a self-reducible problem.

We need only consider $\lambda\geq e/n$ (see Section~\ref{ss:glauber} below). If $\lambda<e/n$ then $1\leq P_G(\lambda)< e^e < 16$, using~\eqref{eq:Pbounds}.
In this case, it suffices to sample independent sets of size at most
$k=O(1)$ according to~\eqref{Gibbs}, as larger independent sets will have
negligible stationary probability.
This can be done in
$O(n^k)$ time by enumerating all independent sets of size at most~$k$. This is polynomial for constant $k$, though
counting is \#W[1]-hard viewed as a fixed parameter problem, even for line graphs~\cite{Curt}. We omit the details here.

Therefore we can assume from now on that $\lambda \geq e/n$.
Under this assumption, \eqref{eq:Pbounds} can be tightened to
\begin{equation}\label{eq:Pbounds1}
P_G(\lambda)\,\leq\, \sum_{k=0}^\alpha \binom{n}{k}\, \lambda^k\,\leq\, \sum_{k=0}^\alpha\frac{(n\lambda)^k}{k!}
\,\leq\, (n\lambda)^\alpha\sum_{k=0}^\alpha\frac{1}{k!}\,\leq \,e(n\lambda)^\alpha.
\end{equation}
%

\section{Markov chains}\label{s:markov}

For additional information on Markov chains and approximate counting,
see for example~\cite{jerrumbook}.  In this section we provide some
necessary definitions and then define a
simple Markov chain on
the set of independent sets in a graph.

\subsection{Mixing time}\label{s:mixing}

Consider a Markov chain on state space $\Omega$ with stationary distribution
$\pi$ and transition matrix $\tprob$.
Let $p_n$ be the distribution of the chain after $n$ steps.
We will assume that $p_0$ is the distribution which assigns probability~1 to
a fixed initial state $x\in\Omega$.
The \emph{mixing time} of the Markov chain, from initial state $x\in\Omega$, is
\[
\tau_x(\eps)= \min\{n\, :\, \dtv(p_n,\pi)\leq\eps\},
\]
where
\[ \dtv(p_n,\pi) = \frac12\sum_{Z\in\Omega}|p_n(Z)-\pi(Z)|\]
is the \emph{total variation distance} between $p_n$ and $\pi$.

In the case of the Glauber dynamics for independent sets, the stationary
distribution $\pi$ satisfies (\ref{Gibbs}), and in particular
$\pi(\es)^{-1}=P_G(\lambda)$. We will always use $\es$ as our starting state, as
it is an independent set in every graph.

Let $\beta_{\max} = \max\{\beta_1,|\beta_{|\Omega|-1}|\}$, where $\beta_1$ is the second-largest
eigenvalue of $\tprob$ and $\beta_{|\Omega|-1}$ is the smallest eigenvalue of $\tprob$.
It follows from~\cite[Proposition 3]{DS91} that
\[ \tau_{x}(\eps) \leq (1-\beta_{\max})^{-1}\, \left(\ln(\pi(x)^{-1}) + \ln(1/\eps)\right),
\]
see also~\cite[Proposition~1(i)]{sinclair}.
Therefore, if $\lambda \geq 1/n$ then
\begin{equation}
\label{both-eigvals}
 \tau_{\es}(\eps) \leq (1-\beta_{\max})^{-1}\, \left(\alpha(G)\ln(n\lambda)) + 1 +  \ln(1/\eps)\right),
\end{equation}
using (\ref{eq:Pbounds1}).

We can easily prove that $(1+\beta_{|\Omega|-1})^{-1}$
is bounded above by a $\min\{\lambda, n\}$, see (\ref{smallest-eigval}) below.
It is more difficult to bound the \emph{relaxation time}
$(1-\beta_1)^{-1}$. We use the
canonical paths method, which we now describe, for this task.

\subsection{Canonical paths method}\label{ss:canonical}

To bound the mixing time of our Markov chain
we will apply the \emph{canonical paths} method of Jerrum and Sinclair~\cite{JS}.
This may be summarised as follows.

Let the problem size be $n$ (in our setting, $n$ is the number of vertices
in the graph $G$ and $|\cI(G)|\leq 2^n$).
For each pair of states $X, Y\in\Omega$ we must define a path
$\gamma_{XY}$ from $X$ to $Y$,
\[ X=Z_0\to Z_2\to\cdots\to Z_\ell=Y\]
such that successive pairs along the path are given by a transition of
the Markov chain.  Write $\ell_{XY}=\ell$ for the length of the path
 $\gamma_{XY}$, and let $\lmax=\max_{X,Y}\ell_{XY}$.
We require $\lmax$ to be at most polynomial in $n$. This is usually easy to
achieve, but the set of paths $\{\gamma_{XY}\}$ must also satisfy
the following, more demanding property.

For any transition $(Z,Z')$  of the chain there must exist an
\emph{encoding} $W$, such that, given $(Z,Z')$ and $W$,
there are at most $\nu$ distinct possibilities for $X$ and $Y$
such that $(Z,Z')\in \gamma_{XY}$.   That is, each transition of the
chain can lie on at most $\nu\, |\Omega^\ast|$ canonical paths,
where $\Omega^\ast$ is some set which contains all possible encodings.
We usually require $\nu$ to be polynomial in $n$.
It is common to refer to the additional information provided by $\nu$ as
``guesses'', and we will do so here.
In our situation, all encodings will be independent sets, so we may
assume that $\Omega^\ast=\Omega=\cI(G)$.
{Furthermore, independent sets are weighted by $\lambda$, so we will
need to perform a weighted sum over our ``guesses''.  See the proof of
Theorem~\ref{thm:main-Cp} in Section~\ref{s:ispaths}.}

The \emph{congestion}~$\varrho$ of the chosen set of paths is given by
\begin{equation}\label{eq:congestion}
\varrho=\max_{(Z,Z')}\bigg\{\frac{1}{\pi(Z)\tprob(Z,Z')}
\sum_{X,Y:\gamma_{XY}\ni(Z,Z')}\pi(X)\pi(Y)\,\bigg\},
\end{equation}
where the maximum is taken over all pairs $(Z,Z')$
with $\tprob(Z,Z')>0$ and $Z'\neq Z$ (that is, over all transitions
of the chain), and the sum is over all
paths containing the transition $(Z,Z')$.

A bound on the relaxation time $(1-\beta_1)^{-1}$ will follow from a bound on congestion, using
Sinclair's result~\cite[Cor.~6]{sinclair}:
\begin{equation}
\label{eq:trel}
(1-\beta_1)^{-1} \leq \lmax\, \varrho.
\end{equation}

\subsection{Glauber dynamics}\label{ss:glauber}

The Markov chain we employ will be the \emph{Glauber dynamics} on state
space $\Omega=\cI(G)$.
In fact, we will consider a weighted version of this chain, for a given
value of the fugacity (also called activity) $\lambda>0$.
Define $\pi(Z) = \lambda^{|Z|}/P_G(\lambda)$ for all $Z\in \cI(G)$,
where $P_G(\lambda)$ is the independence polynomial defined in (\ref{independence}).
A transition from $Z\in\cI(G)$ to $Z'\in\cI(G)$ will be as follows.
Choose a vertex $v$ of $G$ uniformly at random.
\begin{itemize}
\item If $v\in Z$ then $Z'\gets Z\sm\{v\}$ with probability $1/(1+\lambda)$.
\item If $v\notin Z$ and $Z\cup\{v\}\in\cI(G)$ then  $Z'\gets Z\cup\{v\}$
   with probability $\lambda/(1+\lambda)$.
\item Otherwise $Z'\gets Z$.
\end{itemize}
This Markov chain is irreducible and aperiodic,
and satisfies the detailed balance equations
\[ \pi(Z)\, \tprob(Z,Z')=\pi(Z')\, \tprob(Z',Z)\]
for all $Z,Z'\in \cI(G)$.
Therefore, the Gibbs distribution $\pi$ is the
stationary distribution of the chain.
Indeed, if $Z'$ is obtained from $Z$ by deleting a vertex $v$ then
\begin{equation} \tprob(Z,Z')=\frac{1}{n(1+\lambda)} \quad \text{ and } \quad
   \tprob(Z',Z)=\frac{\lambda}{n(1+\lambda)}.
\label{transitions}
\end{equation}

Very recently, Chen et al.~\cite{clv} proved that the Glauber dynamics converges in
time $O(n^{2+32/\delta})$ whenever $\lambda \leq (1-\delta) \lambda_c$,
where $\Delta$ is the maximum vertex degree and
$\lambda_c = (\Delta-1)^{\Delta-1}/(\Delta-2)^\Delta$ is the tree uniqueness
threshold. 
Now, by an easy calculation, $\lambda<e/n$ implies that $\lambda<\lambda_c$.  
Therefore, as claimed in Section~\ref{ss:notation},
we may assume that $\lambda\geq e/n$. 

The unweighted version is given by setting $\lambda=1$, and has
uniform stationary distribution. Since the analysis for
general $\lambda$ is hardly any more complicated than that for
$\lambda=1$, we will work with the weighted case.

It follows from the transition procedure that $\tprob(Z,Z)\geq \min\{1,\lambda\}/(1+\lambda)$
for all states $Z\in\cI(G)$.  That is, every state has a self-loop probability of at least
this value.  Using a result of Diaconis and Saloff-Coste~\cite[p.~702]{DS93}, we
conclude that the smallest eigenvalue $\beta_{|\cI(G)|-1}$ of $\tprob$ satisfies
\begin{equation}
\label{smallest-eigval}
 (1+\beta_{|\cI(G)|-1})^{-1} \leq \frac{1+\lambda}{2\min\{1,\lambda\}}
 \leq \min\{\lambda, n\}.
\end{equation}
This bound will be dominated by our bound on the relaxation time.
As mentioned earlier, we always use the initial state $Z_0=\es$.

In order to bound the relaxation time
$(1-\beta_1)^{-1}$ we will use the canonical path method.
A key observation is that for any $X,Y\in\cI(G)$, the induced subgraph
$G[X\sd Y]$ of $G$ is bipartite. This can easily be seen by colouring
vertices in $X\sm Y$ black and vertices in $Y\sm X$ white, and
observing that no edge in $G$ can connect vertices of the same colour.
To exploit this observation, we introduce the \emph{bipartite pathwidth}
of a graph in Section~\ref{s:bippathwidth}.
In Section~\ref{s:ispaths} we show how to use the bipartite pathwidth
to construct canonical paths for independent sets, and analyse the
congestion of this set of paths to prove our main result, Theorem~\ref{thm:main-Cp}.

\section{Pathwidth and bipartite pathwidth}\label{s:bippathwidth}

The \emph{pathwidth} of a graph was defined by Robertson and
Seymour~\cite{RS}, and has proved a very useful notion in graph
theory. See, for example,~\cite{HLB:arb,diestel}.
A \emph{path decomposition} of a graph $G=(V,E)$ is a sequence
$\cB = (B_1, B_2, \dots, B_r)$ of subsets of $V$ such that
\begin{enumerate}[itemsep=0pt,topsep=5pt,label=(\roman*)]
\item \label{pw1} for every $v \in V$ there is some $i \in [r]$ such that $v \in B_i$,
\item \label{pw2} for every $e \in E$ there is some $i \in [r]$ such that $e \subseteq B_i$, and
\item \label{pw3} for every $v \in V$ the set $\{i \in [r] \, :\,  v \in B_i\}$ forms an interval in $[r]$.
\end{enumerate}
The \emph{width} and \emph{length} of this path decomposition $\cB$ are
\begin{equation*}
  w(\cB)    = \max\{|B_i| \, :\,  i \in [r]\} - 1, \qquad
  \ell(\cB) = r
\end{equation*}
and the \emph{pathwidth} $\pw(G)$ of a given graph $G$ is
\[ \pw(G) = \min_{\cB} w(\cB) \]
where the minimum taken over all path decompositions $\cB$ of $G$.

Condition \ref{pw3} is equivalent to $B_i \cap B_k \subseteq B_j$ for
all $i$, $j$ and $k$ with $1 \le i \le j \le k \le r$.
If we refer to a bag with index $i \notin [r]$ then by default $B_i = \es$.

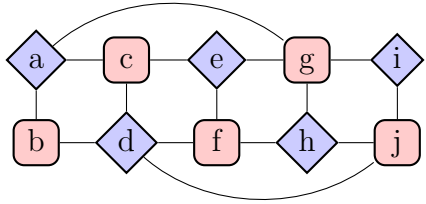
\begin{figure}[htbp]
  \newcommand{\II}[1]{\rule[-0.45ex]{0pt}{2.05ex}#1\rule[-0.45ex]{0pt}{2.05ex}}
  \centering
  \begin{tikzpicture}[xscale=1.2,bend angle=45]
    \foreach[count=\x] \n/\c in {a/B, c/R, e/B, g/R, i/B} \node[\c] (\n) at (\x,2) {\II{\n}};
    \foreach[count=\x] \n/\c in {b/R, d/B, f/R, h/B, j/R} \node[\c] (\n) at (\x,0.9) {\II{\n}};
    \draw (d)--(b)--(a)--(c)--(d)--(f)--(e)--(g)--(h)--(j)--(i)--(g)  (c)--(e)  (f)--(h);
    \draw (a) edge[bend left] (g)  (d) edge[bend right] (j);
  \end{tikzpicture}
  \caption{A bipartite graph}
  \label{fig:a-j}
\end{figure}

For example, the bipartite graph $G$ in Fig.~\ref{fig:a-j} has
a path decomposition with the following bags:

\newcommand{\ivset}[4]{\{\mathrm{#1}, \mathrm{#2}, \mathrm{#3}, \mathrm{#4}\}}
\begin{equation}
\label{bags}
\begin{aligned}
  B_1 &= {\ivset abdg} & B_2 &= \ivset acdg & B_3 &= \ivset cdge & B_4 &= \ivset defg \\
  B_5 &= \ivset dfgj & B_6 &= \ivset fghj & B_7 &= \ivset ghij
\end{aligned}
\end{equation}
This path decomposition has length~7 and width~3, so $\pw(G)\leq 3$.

If $P$ is a path, $C$ is a cycle and $K_{a,b}$ is a
complete bipartite graph, then it is easy to show that
\begin{equation}\label{eq:pw}
\pw(P) = 1,\qquad \pw(C) = 2,\qquad \pw(K_{a,b}) = \min\{a,b\} \,.
\end{equation}

It is well-known that the clique number of a graph, minus 1, is a lower bound
for the pathwidth.  (Indeed, this follows from the corresponding result
about treewidth.)  In particular, the complete graph $K_n$ satisfies
\begin{equation}\label{clique}
\pw(K_n) \geq n - 1.
\end{equation}
(For an upper bound, take a single bag which contains all $n$ vertices.)

The following result will be useful for bounding the pathwidth. The
first statement is~\cite[Lemma~11]{HLB:pw}, while the second appears,
without proof, in~\cite[equation~(1.5)]{RS}. We give a proof of both
statements, for completeness.

\begin{lemma} \label{l:pw is monotone}
  Let $H$ be a subgraph of a graph $G$ (not necessarily an induced subgraph).
Then $\pw(H) \le \pw(G)$.
  Furthermore, if $W \subseteq V(G)$ then $\pw(G) \le \pw(G \sm W) + |W|$.
\end{lemma}

\begin{proof}
  Let $G=(V,E)$ and $H=(U,F)$. Since $H$ is a subgraph of $G$ we have
  $U \subseteq V$ and $F \subseteq E$. Let $(B_i)_{i=1}^r$ be a path
  decomposition of width $\pw(G)$ for $G$. Then $(B_i \cap U)_{i=1}^r$
  is a path decomposition of $H$, and its width is at most $\pw(G)$.

  Now given $W\subseteq V(G)$, let $U=V\setminus W$ and consider the
induced subgraph $H=G[U]$.
  We show) that $\pw(G) \le \pw(H) + |W|$, as follows. Let $(A_i)_{i=1}^s$ be a
  path decomposition of width $\pw(H)$ for $H$. Then $(A_i \cup W)_{i=1}^s$
  is a path decomposition of $G$, and its width is $\pw(H)+|W|$.
This concludes the proof, as $H=G\setminus W$.
\end{proof}
Another helpful property of pathwidth is that, if $H$ is a \emph{minor} of $G$, then $\pw(H)\leq\pw(G)$ \cite[Lem.~16]{HLB:pw}.
Here $H$ is minor of $G$ if it can be obtained from $G$ by deleting vertices
and edges and contracting edges.  We can use this fact
to determine the pathwidth of the graph $G$ in Fig.~\ref{fig:a-j}.
Contracting edges $ac$, $ab$, $bd$, $gi$, $ij$, $jh$, and deleting parallel edges, results in $H\cong K_4$.
Then $\pw(H)=3$, by (\ref{clique}).
So $\pw(G)\geq \pw(H) = 3$, but the path decomposition given in (\ref{bags})
shows that
$\pw(G)\leq 3$, and therefore $\pw(G)= 3$.

\subsection{Bipartite pathwidth}\label{ss:bipw}

We now define the \emph{bipartite pathwidth} $\bpw(G)$ of a graph $G$ to be
the maximum pathwidth of an induced subgraph of $G$ that is bipartite.
For any positive integer $p \ge 2$, let $\cC_p$ be the class of graphs
of bipartite pathwidth at most $p$.
Lemma~\ref{lem:claw} below implies that claw-free graphs are contained in
$\cC_2$, for example. Note that $\cC_p$ is a hereditary class,
by Lemma~\ref{l:pw is monotone}.

Clearly $\bpw(G)\leq\pw(G)$, but the bipartite pathwidth of $G$ may be much
smaller than its pathwidth. For example, consider the complete graph $K_n$.
From (\ref{clique}) we have $\pw(K_n)=n-1$, but the largest induced bipartite
subgraphs of $K_n$ are its edges, which are all isomorphic to $K_2$.
Thus the bipartite pathwidth of $K_n$ is $\pw(K_2)=1$.

A more general example is the class of \emph{unit interval graphs}.
These may have cliques of arbitrary size, and hence arbitrary pathwidth.
However they are claw-free, so their induced bipartite subgraphs are
linear forests (forests of paths), and hence by (\ref{eq:pw}) they have bipartite pathwidth at most~$1$.
The even more general \emph{interval graphs} do not
contain a tripod (depicted in Figure~\ref{fig:tas}), so their bipartite
subgraphs are forests of caterpillars, and hence they have bipartite pathwidth
at most~$2$.

We also note the following.
\begin{lemma}
\label{lem:notes}
Let $p$ be a positive integer.
\vspace*{-0.5\baselineskip}
\begin{itemize}
  \item[\emph{(i)}] Every graph with at most $2p+1$ vertices belongs to $\cC_p$.
  \item[\emph{(ii)}] No element of $\cC_p$ can contain $K_{p+1,p+1}$ as an induced subgraph.
\end{itemize}
\vspace*{-0.5\baselineskip}
\end{lemma}

\begin{proof}
  Suppose that $G$ has $n$ vertices, where $n\leq 2p+1$.
  Let $H=(X,Y,F)$ be a bipartite induced subgraph of $G$ such that $\bpw(G)=\pw(H)$.
  Since $n \le 2p+1$ we have $|X| \le p$ or $|Y| \le p$.
  That is, $H$ is a subgraph of $K_{p,n-p}$.
  By Lemma~\ref{l:pw is monotone} we have
  \[ \bpw(G)=\pw(H) \le \pw(K_{p,n-p}) \le p,\]
  proving (i).
  For (ii) suppose that a graph $G'$ contains $K_{p+1,p+1}$ as
  an induced subgraph.
  Then $\bpw(G')\geq\pw(K_{p+1,p+1}) = p+1$, from \eqref{eq:pw}.
\end{proof}

We say that a path decomposition $(B_i)_{i=1}^r$ is \emph{good} if, for all $i
\in [r-1]$, neither $B_i \subseteq B_{i+1}$ nor $B_i \supseteq
B_{i+1}$ holds. Every path decomposition of $G$ can be transformed
into a good one by leaving out any bag which is contained in another.

It will be useful to define a partial order on path decompositions.
Given a fixed linear order on the vertex set $V$ of a graph $G$,
we may extend $<$ to subsets of $V$ as follows:
if $A,B \subseteq V$ then $A < B$ if and only if
(a) $|A|< |B|$; or
(b) $|A|=|B|$ and the smallest element of $A \sd B$ belongs to $A$.
Next, given two path decompositions
$\cA = (A_j)_{j=1}^r$ and  $\cB = (B_j)_{j=1}^s$ of $G$, we say that
$\cA < \cB$ if and only if (a) $r<s$; or (b) $r=s$ and $A_j < B_j$,
where $j = \min\{ i \, : \, A_i \neq  B_i\}$.

\section{Canonical paths for independent sets}\label{s:ispaths}

We now construct canonical paths for the Glauber dynamics on independent sets
of graphs with bounded bipartite pathwidth.

Suppose that $G\in\cC_p$, so that $\bpw(G)\leq p$.
Take $X,Y\in \cI(G)$ and let $H_1,\ldots, H_t$ be the connected
components of $G[X\sd Y]$, ordered in lexicographical order.
As already observed, the graph $G[X\sd Y]$ is bipartite, so every
component $H_1,\ldots, H_t$ is connected and bipartite.
We will define a canonical path $\gamma_{XY}$ from $X$ to $Y$
by processing the components $H_1,\ldots, H_t$ in order.

Let $H_a$ be the component of $G[X\sd Y]$ which we are currently processing,
and suppose that after processing $H_1,\ldots, H_{a-1}$ we have a
partial canonical path
\[ X = Z_0,\ldots, Z_{N}.  \]
If $a=0$ then $Z_N = Z_0=X$.

The encoding $W_N$ for $Z_N$ is defined by
\begin{equation}
\label{perfect}
   Z_N\sd W_N = X\sd Y \quad \text{ and } \quad Z_N\cap W_N = X\cap Y.
\end{equation}
In particular, when $a=0$ we have $W_0=Y$.
{We remark that (\ref{perfect}) will not hold during the processing of
a component, but always holds immediately after the processing of a
component is complete.
}
Because we process components one-by-one, in order, and due to the definition of the
encoding $W_N$, we have
\begin{equation}
\label{processed}
\left.
\begin{aligned}
 Z_N\cap H_s = \begin{cases} Y\cap H_s & \text{ for $s=1,\ldots, a-1$ (processed),}\\
    X\cap H_s & \text{ for $s=a,\ldots, t$ (not processed),} \end{cases}\\
 W_N\cap H_s = \begin{cases} X\cap H_s & \text{ for $s=1,\ldots, a-1$ (processed),}\\
    Y\cap H_s & \text{ for $s=a,\ldots, t$ (not processed).}
\end{cases}
\end{aligned}
\,\, \right\}
\end{equation}

We now describe how to extend this partial
canonical path by processing the component $H_a$.
Let $h=|H_a|$.
While processing $H_a$ we will produce a sequence
\begin{equation}
\label{path}
   Z_{N},\, Z_{N+1},\ldots , Z_{N + h}
\end{equation}
of independent sets, and a corresponding sequence
\[ W_N,\, W_{N+1},\ldots, W_{N+h}
\]
of encodings.  We will ensure that
\[ Z_{\ell}\sd W_{\ell} \subseteq X\sd Y \quad \text{ and } \quad
   Z_{\ell}\cap W_{\ell} = X\cap Y\]
for $j=N,\ldots, N+h$.  
When we reach $Z_{N+h}$, the processing of $H_a$
is complete and (\ref{perfect}) holds with $N$ replaced by $N+h$.
Define the set of ``remembered vertices''
\[  R_{\ell} = (X\sd Y)\setminus (Z_{\ell}\sd W_{\ell})\]
for $\ell=N,\ldots, N+h$.  By definition,
the triple $(Z,W,R) = (Z_{\ell},W_{\ell},R_{\ell})$
satisfies
\begin{equation}
\label{remember}
  (Z\sd W)\cap R = \es \quad \text{ and } \quad (Z\sd W)\cup R = X\sd Y.
\end{equation}
This immediately implies that $|Z_\ell|+|W_\ell| + |R_\ell|=|X|+|Y|$
for $\ell=N,\ldots, N+h$.

We use a path decomposition of $H_a$ to guide our construction of the
canonical path.
Let $\cB=(B_1,\ldots, B_r)$ be the lexicographically-least good
path decomposition of $H_a$.   Here we use the ordering on path decompositions
defined at the end of Section~\ref{ss:bipw}.
Since $G\in \cC_p$, the maximum bag size in $\cB$ is
$d\leq p+1$. As usual, we assume that $B_0,B_{r+1}=\es$.

We process $H_a$ by processing the bags $B_1,\ldots, B_r$ in order.
Initially $R_N = \es$, by (\ref{perfect}).
Because we process the bags one-by-one, in order, if bag $B_i$ is currently being
processed and the
current independent set is $Z$ and the current encoding is $W$, we will ensure
that
\begin{equation}
\label{bag-processed}
\left.
\begin{aligned}
\big(X \cap (B_1\cup \cdots \cup B_{i-1})\big)\sm B_i  &= \big(W\cap (B_1\cup\cdots\cup B_{i-1})\big)\sm B_i ,\\
\big(Y \cap (B_1\cup \cdots \cup B_{i-1})\big)\sm B_i  &= \big(Z\cap (B_1\cup\cdots\cup B_{i-1})\big)\sm B_i ,\\
\big(X \cap (B_{i+1}\cup \cdots \cup B_r)\big)\sm B_{i} &= \big(Z\cap (B_{i+1}\cup\cdots\cup B_r)\big)\sm B_{i},\\
\big(Y \cap (B_{i+1}\cup \cdots \cup B_r)\big)\sm B_{i} &= \big(W\cap (B_{i+1}\cup\cdots\cup B_r)\big)\sm B_{i}.\\
\end{aligned}
\,\,\,\right\}
\end{equation}

It remains to describe how to process the bag $B_i$, for $i=1,\ldots, r$.
First we give a high-level overview, and then a more detailed description.

Let $Z_{\ell}$, $W_{\ell}$, $R_{\ell}$
 denote the current independent set, encoding and set of remembered vertices,
immediately after the processing of bag $B_{i-1}$.
So $\ell = N + j$ for some $j\geq 0$, where the processing of bags $B_1,\ldots, B_{i-1}$
required $j$ steps.
From $Z_{\ell}$, we would like to delete vertices of 
$B_i$ which do not belong to $Y$, 
one by one, and then insert vertices of $B_i$ which belong to $Y$, one by one, to make 
$B_i$ agree with $Y$ by the end
of its processing.   When we delete a vertex from the current independent set
we would like to insert it into the corresponding encoding, and vice-versa.
However, we also want to ensure that the current encoding is an independent set 
at each step. 
Two problems can occur:
\begin{itemize}
\item Sometimes we cannot insert a vertex $u$ into the current independent set,
because a neighbour $w$ of $u$ is present, which will be deleted when a later bag
is processed.  We will remember $u$ by adding it to the set of remembered vertices,
and insert it into the independent set when it is safe to do so.
\item Sometimes when a vertex $u$ is deleted from the current independent set,
we cannot immediately add it to the current encoding, because some neighbours of
$u$ are still present in the current encoding. These neighbours will be inserted
into the independent set (and deleted from the encoding) when a later bag is
processed.  We will remember $u$ by adding it to the set of remembered vertices,
and insert it into the encoding when it is safe to do so.
\end{itemize}
In both of these problem cases, vertex $u$ must also belong to $B_{i+1}$, 
This follows from the definition of a path decomposition,  
since $u$ has at least one neighbour outside $B_1\cup\cdots B_i$ which has not yet been
processed.
We handle these two types of remembered vertices in a ``preprocessing phase'' and 
``postprocessing phase'' which occur before and after the deletion and insertion
steps, respectively. 

Now we give a more detailed description.
We write
\[ R_\ell = R_{\ell}^+ \cup R_{\ell}^-\]
where vertices in $R_{\ell}^+$ are added to $R_\ell$ during the preprocessing
phase (and must eventually be inserted into the current independent set),
and vertices in $R_{\ell}^-$ are added to $R_\ell$ due to a deletion step
(and will go into the encoding during the postprocessing phase).
When $i=0$ we have $\ell = N$ and in particular, $R_N = R_N^+ = R_N^- =\es$.
\begin{enumerate}[itemsep=5pt,topsep=5pt,label=(\arabic*)]
\item \emph{Preprocessing}:\, \textit{We ``forget'' the vertices of $B_i\cap B_{i+1}\cap W_{\ell}$ and add them to $R_{\ell}^+$. \\
    This does not change the current independent set or add to the canonical path.} \\
  \begin{minipage}{\textwidth}
    \begin{tabbing}
       \textbf{for} \=\kill
      \textbf{begin} \\
      \>$R_{\ell}^+ \gets R_{\ell}^+ \cup (B_i\cap B_{i+1}\cap W_{\ell})$; \\
      \>$W_{\ell}   \gets W_{\ell} \sm (B_i\cap B_{i+1})$; \\
      \textbf{end}
    \end{tabbing}
  \end{minipage}
\item \emph{Deletion steps}: \\
  \begin{minipage}{\textwidth}
    \begin{tabbing}
      \textbf{for} \=each $u\in B_i\cap Z_{\ell}$, in lexicographical order,
      \textbf{do}\\
        \>$Z_{\ell+1}\gets Z_\ell\sm \{ u\}$;\\
        \>\textbf{if} $u\not\in B_{i+1}$ \=\textbf{then}
          \=$W_{\ell+1}\gets W_\ell \cup \{u\}$;
          \=$R^-_{\ell+1}\gets R^-_{\ell}$; \\
        \>\>\textbf{else}
          \>$W_{\ell+1}\gets W_\ell$;
          \>$R^-_{\ell+1}\gets R^-_\ell\cup\{u\}$; \\
        \>\textbf{end if} \\
        \> $\ell\gets\ell+1$; \\
      \textbf{end do}
    \end{tabbing}
  \end{minipage}

\item \emph{Insertion steps}: \\
  \begin{minipage}{\textwidth}
    \begin{tabbing}
      \textbf{for} \=each $u\in \big(B_i\cap (W_{\ell}\cup R_{\ell}^+)\big)
      \sm B_{i+1}$, in lexicographic order, \textbf{do} \\
        \>$Z_{\ell+1}\gets Z_\ell\cup\{u\}$;\\
        \>\textbf{if} $u\in W_{\ell}$ \=\textbf{then}
          \=$W_{\ell+1}\gets W_\ell \sm \{u\}$;
          \=$R^+_{\ell+1}\gets R^+_{\ell}$; \\
        \>\>\textbf{else}
          \>$W_{\ell+1}\gets W_\ell$;
          \>$R^+_{\ell+1}\gets R^+_\ell\cup\{u\}$; \\
        \>\textbf{end if} \\
     \>$\ell\gets\ell+1$; \\
     \textbf{end do}
   \end{tabbing}
  \end{minipage}

\item \emph{Postprocessing:}
    \textit{Any elements of $R^-_{\ell+1}$ which do not belong to $B_{i+1}$ can now be safely added to $W_{\ell}$.
    This does not change the current independent set or add to the canonical path.} \\
  \begin{minipage}{\textwidth}
    \begin{tabbing}
      \textbf{for} \=\kill
      \textbf{begin} \\
      \>$W_{\ell}   \gets W_{\ell} \cup (R_{\ell}^- \sm B_{i+1})$; \\
      \>$R_{\ell}^- \gets R_{\ell}^- \cap B_{i+1}$; \\
      \textbf{end}
    \end{tabbing}
  \end{minipage}
\end{enumerate}
By construction, vertices added to $R_{\ell}^+$ are removed from
$W_{\ell}$, so the ``else'' case for insertion is precisely $u\in R_{\ell}^+$.

Observe that both $Z_\ell$ and $W_\ell$ are independent sets at every step.
This is true initially (when $\ell=N$) and remains true by construction.
Indeed, the preprocessing phases removes all vertices
of $B_i\cap B_{i+1}$ from $W_\ell$, which makes more room for other vertices
to be inserted into the encoding later.  A deletion step shrinks the current
independent set and adds the removed vertex into $W_\ell$ or $R_\ell^-$.
A deleted vertex is only added to $R_\ell^-$ if it belongs to $B_i\cap B_{i+1}$,
and so might have a neighbour in $W_{\ell}$.
Finally, in the insertion steps we add vertices from
$\big(B_i\cap (W_\ell\cup R^+_{\ell})\big)\setminus B_{i+1}$
to $Z_\ell$, now that we have made room.  Here $B_i$ is the last bag
which contains the vertex being inserted into the independent set, so any
neighbour of this vertex in $X$ has already been deleted from the current independent
set.  This phase can only shrink the encoding $W_\ell$.

Also observe that (\ref{remember}) holds for $(Z,W,R) = (Z_\ell,W_\ell,R_\ell)$
at every point.  Finally, by construction we have $R_\ell\subseteq B_i$ at
all times.

To give an example of the canonical path construction, we return to the
bipartite graph shown in Figure~\ref{fig:a-j}, which we now treat as the
symmetric difference of two independent sets.
Let $X = \{\mathrm{a},\mathrm{d},\mathrm{e},\mathrm{h},\mathrm{i}\}$ be the set of vertices which are shown as blue diamonds in 
Figure~\ref{fig:a-j} and let $Y = \{\mathrm{b},\mathrm{c},\mathrm{f},\mathrm{g}, \mathrm{j}\}$ be the remaining vertices,
shown as red squares with rounded corners
in Figure~\ref{fig:a-j}.
Table~\ref{tab:ex} illustrates the 10~steps of the canonical path (3 steps to process
bag $B_1$, none to process bag $B_2$, 2 steps to process bag $B_3$, and so on).
In Table~\ref{tab:ex}, blue (diamond) vertices belong to the current independent set $Z$
and red (square) vertices belong to the current encoding $W$.  We only show the vertices
of the bag $B_i$ which is currently being processed, as we can use
(\ref{bag-processed}) for all other vertices.  The white vertices (drawn as circles) are precisely
those which belong to $R$, where elements of $R_{\ell}^-$ are marked ``$-$''.
The column headed ``pre/post processing''  shows the situation directly
\emph{after} the preprocessing phase. Then on the line below, the situation
directly after the postprocessing phase is shown \emph{unless} there is
no change during preprocessing.
During preprocessing and postprocessing, the current independent set does not change, and so these
phases do not contribute to the canonical path.

After processing the last bag, all vertices of $X$ are red (belong to
the final encoding $W$) and all vertices of $Y$ are blue (belong to the final
independent set $Z$), as expected.

\begin{table}[ht!]
  \tikzset{bend angle=45}
  \newcommand{\rb}[1]{\raisebox{2.0mm}{#1}}
  \newcommand{\II}[1]{\rule[-0.45ex]{0pt}{2.05ex}#1\rule[-0.45ex]{0pt}{2.05ex}}
  \centering
  \renewcommand{\arraystretch}{3.0}
  \begin{tabular}{r@{\hspace{3mm}}c@{\hspace{6mm}}c@{\hspace{6mm}}c@{\hspace{6mm}}c}
   $B_i$ &  pre/post processing  & after 1st step & after 2nd step & after 3rd step\\ \hline
     \rb{$B_1$} &
    \begin{tikzpicture}[scale=0.9]
      \foreach[count=\x] \n/\c in {d/B, b/R, a/B, g/U} \node[\c] (\n) at (\x,1) {\II{\n}};
      \draw (d)--(b) -- (a) -- (g);
    \end{tikzpicture} &
    \begin{tikzpicture}[scale=0.9]
      \foreach[count=\x] \n/\c in {d/B, b/R, a/U, g/U} \node[\c] (\n) at (\x,1) {\II{\n}};
      \draw (d)--(b) -- (a) -- (g);
\node[above] at (3,1.2) {$-$};
    \end{tikzpicture} &
    \begin{tikzpicture}[scale=0.9]
      \foreach[count=\x] \n/\c in {d/U, b/R, a/U, g/U} \node[\c] (\n) at (\x,1) {\II{\n}};
      \draw (d)--(b)--(a)--(g);
\node[above] at (1,1.2) {$-$};
\node[above] at (3,1.2) {$-$};
    \end{tikzpicture} &
    \begin{tikzpicture}[scale=0.9]
      \foreach[count=\x] \n/\c in {d/U, b/B, a/U, g/U} \node[\c] (\n) at (\x,1) {\II{\n}};
      \draw (d)--(b)--(a)--(g);
\node[above] at (1,1.2) {$-$};
\node[above] at (3,1.2) {$-$};
    \end{tikzpicture} \\
\hline
    \rb{$B_2$} & \begin{tikzpicture}[scale=0.9]
      \foreach[count=\x] \n/\c in {d/U, c/U, a/U, g/U} \node[\c] (\n) at (\x,1) {\II{\n}};
      \draw (d)--(c)--(a)--(g);
\node[above] at (1,1.2) {$-$};
\node[above] at (3,1.2) {$-$};
    \end{tikzpicture} \\
     & \begin{tikzpicture}[scale=0.9]
      \foreach[count=\x] \n/\c in {d/U, c/U, a/R, g/U} \node[\c] (\n) at (\x,1) {\II{\n}};
      \draw (d)--(c)--(a)--(g);
\node[above] at (1,1.2) {$-$};
    \end{tikzpicture} \\
\hline
    \rb{$B_3$} & \begin{tikzpicture}[scale=0.9]
      \foreach[count=\x] \n/\c in {d/U, c/U, e/B, g/U} \node[\c] (\n) at (\x,1) {\II{\n}};
      \draw (d)--(c)--(e)--(g);
\node[above] at (1,1.2) {$-$};
    \end{tikzpicture} &
    \begin{tikzpicture}[scale=0.9]
      \foreach[count=\x] \n/\c in {d/U, c/U, e/U, g/U} \node[\c] (\n) at (\x,1) {\II{\n}};
      \draw (d)--(c)--(e)--(g);
\node[above] at (1,1.2) {$-$};
\node[above] at (3,1.2) {$-$};
    \end{tikzpicture} &
    \begin{tikzpicture}[scale=0.9]
      \foreach[count=\x] \n/\c in {d/U, c/B, e/U, g/U} \node[\c] (\n) at (\x,1) {\II{\n}};
      \draw (d)--(c)--(e)--(g);
\node[above] at (1,1.2) {$-$};
\node[above] at (3,1.2) {$-$};
    \end{tikzpicture}  \\
\hline
    \rb{$B_4$} & \begin{tikzpicture}[scale=0.9]
      \foreach[count=\x] \n/\c in {d/U, f/U, e/U, g/U} \node[\c] (\n) at (\x,1) {\II{\n}};
      \draw (d)--(f)--(e)--(g);
\node[above] at (1,1.2) {$-$};
\node[above] at (3,1.2) {$-$};
    \end{tikzpicture}  \\
     & \begin{tikzpicture}[scale=0.9]
      \foreach[count=\x] \n/\c in {d/U, f/U, e/R, g/U} \node[\c] (\n) at (\x,1) {\II{\n}};
      \draw (d)--(f)--(e)--(g);
\node[above] at (1,1.2) {$-$};
    \end{tikzpicture}  \\
\hline
    \rb{$B_5$} & \begin{tikzpicture}[scale=0.9]
      \foreach[count=\x] \n/\c in {f/U, d/U, j/U, g/U} \node[\c] (\n) at (\x,1) {\II{\n}};
      \draw (f)--(d)--(j) (g);
\node[above] at (2,1.2) {$-$};
    \end{tikzpicture}  \\
     & \begin{tikzpicture}[scale=0.9]
      \foreach[count=\x] \n/\c in {f/U, d/R, j/U, g/U} \node[\c] (\n) at (\x,1) {\II{\n}};
      \draw (f)--(d)--(j) (g);
    \end{tikzpicture}  \\
\hline
    \rb{$B_6$} & \begin{tikzpicture}[scale=0.9]
      \foreach[count=\x] \n/\c in {f/U, h/B, j/U, g/U} \node[\c] (\n) at (\x,1) {\II{\n}};
      \draw (f)--(h)--(j); \draw (h) edge[bend left] (g);
    \end{tikzpicture} &
    \begin{tikzpicture}[scale=0.9]
      \foreach[count=\x] \n/\c in {f/U, h/U, j/U, g/U} \node[\c] (\n) at (\x,1) {\II{\n}};
      \draw (f)--(h)--(j); \draw (h) edge[bend left] (g);
\node[above] at (2,1.2) {$-$};
    \end{tikzpicture} &
    \begin{tikzpicture}[scale=0.9]
      \foreach[count=\x] \n/\c in {f/B, h/U, j/U, g/U} \node[\c] (\n) at (\x,1) {\II{\n}};
      \draw (f)--(h)--(j); \draw (h) edge[bend left] (g);
\node[above] at (2,1.2) {$-$};
    \end{tikzpicture}  \\
\hline
    \rb{$B_7$} & \begin{tikzpicture}[scale=0.9]
      \foreach[count=\x] \n/\c in {h/U, j/U, i/B, g/U} \node[\c] (\n) at (\x,1) {\II{\n}};
      \draw (h)--(j)--(i)--(g) edge[bend right] (h);
\node[above] at (1,1.2) {$-$};
    \end{tikzpicture} &
    \begin{tikzpicture}[scale=0.9]
      \foreach[count=\x] \n/\c in {h/U, j/U, i/R, g/U} \node[\c] (\n) at (\x,1) {\II{\n}};
      \draw (h)--(j)--(i)--(g) edge[bend right] (h);
\node[above] at (1,1.2) {$-$};
    \end{tikzpicture} &
    \begin{tikzpicture}[scale=0.9]
      \foreach[count=\x] \n/\c in {h/U, j/U, i/R, g/B} \node[\c] (\n) at (\x,1) {\II{\n}};
      \draw (h)--(j)--(i)--(g) edge[bend right] (h);
\node[above] at (1,1.2) {$-$};
    \end{tikzpicture} &
    \begin{tikzpicture}[scale=0.9]
      \foreach[count=\x] \n/\c in {h/U, j/B, i/R, g/B} \node[\c] (\n) at (\x,1) {\II{\n}};
      \draw (h)--(j)--(i)--(g) edge[bend right] (h);
\node[above] at (1,1.2) {$-$};
    \end{tikzpicture}  \\
 &
    \begin{tikzpicture}[scale=0.9]
      \foreach[count=\x] \n/\c in {h/R, j/B, i/R, g/B} \node[\c] (\n) at (\x,1) {\II{\n}};
      \draw (h)--(j)--(i)--(g) edge[bend right] (h);
    \end{tikzpicture}  \\
\hline
  \end{tabular}
  \caption{The steps of the canonical path, processing each bag in order.}
  \label{tab:ex}
\end{table}

\subsection{Analysis of the canonical paths}\label{sec:analysis}

Each step of the canonical path changes the current independent set $Z_i$ by inserting or deleting
exactly one element of $X\sd Y$.
Every vertex of $X\setminus Y$ is removed
from the current independent set at some point, and is never re-inserted, while every vertex of $Y\setminus X$
is inserted into the current independent set once, and is never removed.
Vertices in $X\cap Y$ (respectively $(X\cup Y)^c$) are never altered, and belong to
all (respectively, none) of the independent sets in the canonical path.
Therefore
\begin{equation}
\lmax\leq 2\alpha(G).
\label{lmax-bound}
\end{equation}

Next we provide an upper bound for the number of vertices we need to remember at any particular step.

\begin{lemma}\label{l:claw}
At any transition $(Z,Z')$ which occurs during the processing of bag $B_i$,
the set $R$ of  remembered vertices
satisfies $R\subseteq B_i$, with $|R|\leq p$ unless
$Z\cap B_i = W\cap B_i = \es$.  In this case  $R=B_i$, which gives $|R|\leq p+1$,
and $Z' = Z\cup \{u\}$ for some $u\in B_i$.
\end{lemma}

\begin{proof}
By construction, the set $R$ of remembered vertices satisfies $R\subseteq B_i$
throughout the processing of bag $B_i$.  Hence $|R|\leq |B_i| \leq p+1$.
Now $\cB$ is a good path decomposition, and so $B_i\neq B_{i+1}$,
which implies that $|B_i\cap B_{i+1}|\leq p$.   Therefore, whenever
$R\subseteq B_i\cap B_{i+1}$ we have $|R|\leq p$.

Next suppose that $R=B_i$.  By definition, this means that $Z\cap B_i = W\cap B_i
= \es$, so the transition $(Z,Z')$ is an insertion step which inserts some vertex of
$u$.
\end{proof}

Now we establish the unique reconstruction property of the canonical paths,
given the encoding and set of remembered vertices.

\begin{lemma}
Given a transition $(Z,Z')$, the encoding $W$ of $Z$ and the set $R$
of remembered vertices, we can
uniquely reconstruct $(X,Y)$ with $(Z,Z')\in \gamma_{XY}$.
\label{unique}
\end{lemma}

\begin{proof}
By construction, (\ref{remember}) holds.
This identifies all vertices in $X\cap Y$ and $(X\cup Y)^c$ uniquely.
It also identifies the connected components $H_1,\ldots, H_t$ of $X\sd Y$,
and it remains to decide, for all vertices in $\bigcup_{s=1}^t H_s$, whether they belong to $X$ or $Y$.

Next, the transition $(Z,Z')$ either inserts or deletes some vertex $u$.
This uniquely determines the connected component $H_a$ of $X\sd Y$ which contains $u$.
We can use (\ref{processed}) to identify
$X\cap H_s$ and $Y\cap H_s$ for all $s\neq a$.  It remains to decide which
vertices of $H_a$ belong to $X$ and which belong to $Y$.

Let $B_1,\ldots, B_r$ be the lexicographically-least good path
decomposition of $H_a$, which is well-defined.
If $Z'=Z\cup \{u\}$ (insertion) then $u\in Y\setminus X$ and we are processing the
last bag $B_i$ which contains $u$.  If $Z' = Z\setminus \{u\}$ then
$u\in X\setminus Y$ and we are processing the first bag $B_i$ which contains $u$.
Hence we can uniquely identify the bag $B_i$ which is currently being processed.
We know that bags $B_1,\ldots, B_{i-1}$ have already been processed, and bags
$B_{i+1},\ldots, B_r$ have not yet been processed.  So (\ref{bag-processed}) holds,
which uniquely determines $X$ and $Y$ outside $B_i\cap B_{i+1}$.

Finally, for every vertex $x\in B_i\setminus \{ u\}$, there is a path in $G$ from $x$
to $u$ (the vertex which was inserted or deleted in the given transition).
Since $G[H_a]$ is bipartite and connected, and we have decided for all vertices outside
$B_i\setminus \{u\}$ whether they belong to $X$ or $Y$, it follows that
 we can uniquely reconstruct all of $X\cap (B_i\setminus \{u\})$ and
$Y\cap (B_i\setminus \{u\})$.
This completes the proof.
\end{proof}

We are now able to prove our main theorem, which we restate below.

\mainthm*
\begin{proof}
For a given set $A$, let $\binom{A}{\leq p}$ denote the set of all subsets of $A$ with
at most $p$ elements.
Let $(Z,Z')$ be a given transition of the Glauber dynamics.
To bound the congestion of the transition $(Z,Z')$ we must
sum over all possible encodings $W$ and all possible sets $R$
of remembered vertices.  Here $R$ is disjoint from $Z\sd W$ and  in almost
all cases $|R|\leq p$, by Lemma~\ref{l:claw}. In the exceptional case we have
$|R|\leq p+1$ but we also know the identity of a vertex $u\in R$, since
$u$ is the vertex inserted in the transition $(Z,Z')$. Therefore in all cases,
we only need to ``guess'' (choose) at most $p$ vertices for $R$, from a subset
of at most $n$
vertices.

By Lemma~\ref{unique},
the choice of $(W,R)$ uniquely specifies a pair $(X,Y)$ of independent
sets with $(Z,Z')\in \gamma_{XY}$.  Therefore, using
the stationary distribution $\pi$ defined in (\ref{Gibbs}), and
the assumption that $\lambda\geq e/n$, we have
\begin{align*}
\sum_{X,Y:\gamma_{XY}\ni(Z,Z')}\pi(X)\pi(Y)&= \frac{1}{P_G(\lambda)^2}\, \sum_{X,Y:\gamma_{XY}\ni(Z,Z')}\lambda^{|X|+|Y|}\\
&\leq \frac{1}{P_G(\lambda)^2}\, \sum_{W\in\Omega} \,\,\,
                             \sum_{R\in \binom{V(G)\setminus (Z\sd W)}{\leq p}}\,  \lambda^{|Z|+|W| + |R|}\\
&\leq \frac{\lambda^{|Z|}}{P_G(\lambda)}\, \sum_{W\in\Omega} \frac{\lambda^{|W|}}{P_G(\lambda)}\,
    \, e(n\lambda)^p\\
&= e(n\lambda)^p\, \pi(Z)\,  \sum_{W\in\Omega}\pi(W)\\
&= e(n\lambda)^p\, \pi(Z).
\end{align*}
The second inequality follows as
\[ \sum_{k=0}^p \binom{n}{k}\lambda^k\,  < \,\sum_{k=0}^p \frac{n^k\lambda^k}{k!}\,
< \,(n\lambda)^p\sum_{k=0}^p  \frac{1}{k!} \, <\, e(n\lambda)^p \,,\]
since $p$ is a positive integer and $n\lambda\geq 1$.

Then \eqref{eq:congestion}  gives
\begin{align*}
  \varrho 
\leq 2(n\lambda)^p/\min_{(Z,Z')}\tprob(Z,Z')
 &= 2\, (n\lambda)^p\, n(1+\lambda)/\min\{1,\lambda\}\\
  &=\,
2 n^{p+1}\lambda^p\big(1+\max(\lambda,1/\lambda)\big)\,
\end{align*}
using the transition probabilities from (\ref{transitions}).
%
Combining this with (\ref{eq:trel}) and (\ref{lmax-bound}) gives
\begin{equation}
\label{beta1-bound}
(1-\beta_1)^{-1} \leq  2e\alpha(G)\, n^{p+1}\, \lambda^p
\end{equation}
and the result follows, by~\eqref{both-eigvals} and~\eqref{smallest-eigval}.
\end{proof}

\bigskip

The analysis presented above uses the bound $|R|\leq p+1$ for the
number of remembered vertices.
This bound might seem too large but unfortunately, it can be tight, as we now show.
For any integer $k \ge 1$ let $G=P_4^{(k)}$ be the bipartite graph obtained from a $P_4=
(a,b,c,d)$ by replacing its vertices by independent sets $A$, $B$, $C$, $D$
of size $k$, and its edges by complete bipartite graphs between these sets. See Fig.~\ref{fig:P44} for $P_4^{(4)}$.
\smallskip

\begin{figure}[htbp]
  \newcommand{\II}[1]{\rule[-0.45ex]{0pt}{2.05ex}#1\rule[-0.45ex]{0pt}{2.05ex}}
  \centering
  \begin{tikzpicture}[xscale=3,bend angle=45,scale=0.8]
    \foreach \x in {1,2,3,4} {
    \foreach \y in {1,2,3,4} { \draw (\x,\y) node[b] (\x\y) {};}}
    \foreach \x in {1,2,3,4} {
    \foreach \y in {1,2,3,4} { \draw (1\x)--(2\y) (2\x)--(3\y) (3\x)--(4\y);}}
  \end{tikzpicture}
  \caption{$P_4^{(4)}$}
  \label{fig:P44}
\end{figure}

We will first
show that $\pw(G) = 2k-1$, and that the intersection of
bags in any good path decomposition of $G$ has size $k$.

Now $\pw(G) \le 2k-1$ is shown by the
path decomposition $\Pi=(A \cup B, B \cup C, C \cup D)$.
Observe that $\Pi$  is a good decomposition with the property that the intersection
of each pair of consecutive bags has size~$k$.

To show that $\pw(G)\geq 2k-1$, contract the edges of a matching in $(A,B)$
and a matching in $(C,D)$, and delete parallel edges. This results in a graph
$H\cong K_{2k}$, and hence $\pw(G)\geq \pw(H)=2k-1$, and the path decomposition of $H$
has only one bag.

To show that $\Pi$ is the unique good decomposition, note that
this argument for $\pw(G)\geq 2k-1$ implies that $B\cup C$ must be contained in a bag,
and hence must form a bag, since $|B\cup C|=2k$.
Then it is clear that any path decomposition of $G$
with width $2k-1$ must have at least three bags. Otherwise we must have the bag
$A\cup D$. But then all edges in $G[A\cup B]$, $G[C\cup D]$ would not lie in any bag,
a contradiction. Hence $\Pi$ is the only good decomposition of $G$, and it
has the desired properties.

Observe that the bag $B\cup C$ has size $2k$, and has intersection size $k$ with the first, $A\cup B$, and the third, $C\cup D$. So,
suppose we are constructing a canonical path between the independent sets $A\cup C$ and $B\cup D$.  Then, in processing the bag $B\cup C$ as described above, we must remember all vertices in
$B$ until we have deleted all vertices in $C$. Thus we are required to remember all vertices in $B\cup C$; that is, $R=B\cup C$.

\subsection{Vertex weights}\label{ss:vertexweights}

Jerrum and Sinclair~\cite{JS} extended their algorithm for matchings to the case of positive edge weights,
provided these are uniformly bounded by some polynomial in $n$. Edge weights correspond to vertex weights
in line graphs, and line graphs are claw-free. So we might expect that our algorithm extends to graphs with
polynomially-bounded vertex weights. Here we show that this is the case, and we are able to use a somewhat
simpler approach than that used in~\cite{JS}.

To be specific, suppose that $G=(V,E)$ has real-valued \emph{vertex weights} $w(v)\geq 0$ ($v\in V$).
For our purposes, real weights can be approximated by rationals, which can be reduced to weights in $\IN$
by clearing denominators. Vertices of weight zero can be removed from the 
graph, so we will assume that $w(v)\in\IN$ for all $v\in V$.

The weight of an independent set $I\in\cI_k(G)$ is then defined to
be $w(I)=\prod_{v\in I} w(v)$. The total weight of independent sets of size $k$ in $G$ is given by
$W_k(G)=\sum_{I\in\cI_k}w(I)$. Note that, if $w(v) =1$ for all $v\in V$, the \emph{unweighted} case, then $W_k(G)=N_k(G)$.
Thus counting independent sets is a special case.

We say that a graph $G'$ with vertex weights $w'$ is \emph{equivalent} to a graph 
$G$ with vertex weights $w$ if $W_k(G')=W_k(G)$ for all nonnegative integers 
$k\leq \max\{|V(G)|,\, |V(G')|\}$.
In particular, if $G$ and $G'$ are equivalent then $\alpha(G) = \alpha(G')$.

\subsubsection{Equivalence to the unweighted case}\label{unweighted}

If $G=(V,E)$, vertices $u,v\in V$ are called
\emph{true twins} if $\Nb[u]=\Nb[v]$
and \emph{false twins} if $\Nb(u)=\Nb(v)$.

 So let $G=(V,E)$ be a graph with weights $w: V \to \IN$.
We construct an unweighted ``blown up'' version $G_w$ of $G$ by replacing
each vertex $v$ by a clique of size $w(v)$, and each edge by $vw$ by a complete
bipartite graph.
That is, $G_w = (V',E')$ with $V' = \{v_i : v \in V, i \in [w(v)]\}$ and
different vertices $u_i$ and $v_j$ form an edge in $E'$ if $u=v$ or
$uv \in E$, see Fig.~\ref{fig:blowup}.


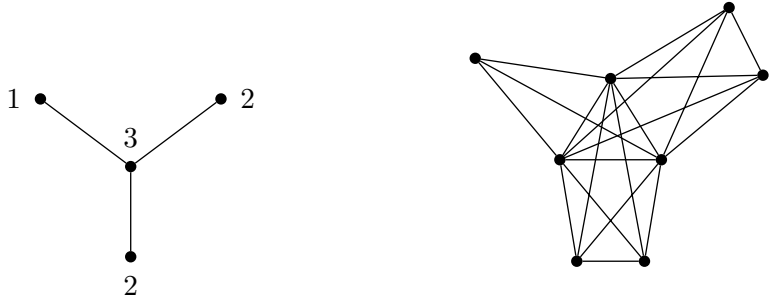
\begin{figure}
\begin{center}
\begin{tikzpicture}[line width=0.5pt,inner sep=5pt,scale=0.6,font=\small]
\node[b,label=below:2] (a1) at (0,0) {} ;
\node[b,label=right:2] (b1) at (2,3.5) {} ;
\node[b,label=above:3] (c1) at (0,2) {} ;
\node[b,label=left:1] (d1) at (-2,3.5) {} ;
\draw (b1)--(c1) ;
\draw (a1)--(c1) ;
\draw (d1)--(c1) ;
\end{tikzpicture}\hspace{1in}
\begin{tikzpicture}[line width=0.5pt,inner sep=5pt,scale=0.9,font=\small]
\node at (0,-.5) {} ;
\node[b] (a1) at (-0.25,0) {} ; \node[b] (a2) at (0.75,0) {} ;
\node[b] (b1) at (2,3.75) {} ; \node[b] (b2) at (2.5,2.75) {} ;
\node[b] (c1) at (-0.5,1.5) {} ; \node[b] (c2) at (1,1.5) {} ; \node[b] (c3) at (0.25,2.7) {} ;
\node[b] (d1) at (-1.75,3) {} ;
\draw (a1)--(a2) (b1)--(b2) (c1)--(c2)--(c3)--(c1) ;
\draw (b1)--(c1)--(b2) (b1)--(c2)--(b2) (b1)--(c3)--(b2) ;
\draw (a1)--(c1)--(a2) (a1)--(c2)--(a2) (a1)--(c3)--(a2) ;
\draw (d1)--(c1) (d1)--(c2) (d1)--(c3) ;
\end{tikzpicture}
\end{center}
\caption{expansion}\label{fig:blowup}
\end{figure}
It is clear that any independent set in $G_w$ contains, for each $v \in V$,
at most one vertex $v_i$, and two vertices $u_i,v_j$ only if $uv\notin E$.
Thus every independent set $S$ in $G$ corresponds to exactly $w(S)$
independent sets in $G_w$. Therefore $N_k(G_w)=W_k(G)$ for all $k\in\IN$.
So the unweighted graph $G_w$ is equivalent to $G$ with weights~$w$.

The transformation from $(G,w)$ to $G_w$ has the following useful property.
\begin{lemma}\label{lem:subgraph}
  Let $\cF$ be a set of graphs, none of which contains true twins.
  Given weights $w:V\to\IN$, the graph $G=(V,E)$  is $\cF$-free if and  only
  if $G_w$ is $\cF$-free.
\end{lemma}
\begin{proof}
  First let $H \in \cF$ be a graph and $U \subseteq V$ such that $G[U]\cong H$.
  If $U_1 = \{u_1 : u \in U\}$, then clearly $G_w[U_1]\cong G[U]\cong H$.

  Now let $U' \subseteq V'$ be such that $H'=G_w[U']\cong H$ for
  $H \in \cF$. Suppose that $v_i,v_j\in V[H']$
  for some $v\in V$ and $i,j\in [w(v)]$.
  By the construction of $G_w$, $\Nb[v_1]=\Nb[v_2]$, so $v_i,v_j$
  are true twins in $H'$, contradicting $H'\cong H$.
  It follows that, for each $v \in V$, there is at most one index
  $i \in [w(v)]$ such that $v_i \in U'$. Therefore the set
  $U= \{v\in V: v_i\in U',i \in [w(v)]\}$
  induces a subgraph $G[U]\cong H$.
\end{proof}

We say that a graph class $\cC$ is \emph{expandable} if $G\in \cC$
implies $G_w\in\cC$ for all weight functions $w\in\IN^n$.
Thus Lemma \ref{lem:subgraph} gives a sufficient condition for $\cC$ to be expandable.
We now show that this sufficient condition is also necessary.
\begin{lemma}\label{lem:2-module}
  Let $\cC$ be a hereditary class of graphs and let $\cF$ be the set of
  minimal forbidden subgraphs of $\cC$. For every graph $H \in \cF$ that
  contains true twins, there exists a graph $G =(V,E)\in \cC$
  and weights $w:V\to\IN$ such that $G_w$ is isomorphic to $H$,
  and hence $G_w \notin \cC$.
\end{lemma}
\begin{proof}
  Let $H=(U,F)$ and let $u_1$ and $u_2$ be true twins in $H$.
  We define $G=H \sm u_2$, $w(u_1)=2$
  and $w(v)=1$ for all $v \in G\sm u_1$. By minimality of $H$ we must
  have $G \in \cC$, and the weights are fixed such that $G_w \cong H$.
\end{proof}
Our main application here is to the class of graphs $\cC_p$ with bipartite pathwidth at most $p\geq 1$. The set $\cF$ of minimal excluded subgraphs for $\cC_p$ is a subset of the set of all connected bipartite graphs with pathwidth at least $p+1$. The only connected bipartite graph that contains true twins
is a single edge, which has pathwidth 1, so is not in $\cF$. Thus $\cC_p$ is expandable. It follows that all our results for $\cC_p$ carry
over to polynomially-bounded vertex weights. Since we show in Lemma~\ref{lem:C_2} that the claw-free graphs
considered in Section~\ref{s:claw-free} are a subclass of $\cC_2$, these are expandable. This can also be seen directly from Lemma~\ref{lem:subgraph}.

Many other hereditary graph classes are characterized by a set $\cF$ of forbidden
induced subgraphs that meets the condition of Lemma~\ref{lem:subgraph}.
Among them are fork-free graphs and \fast graphs considered
in Section~\ref{s:classes}. Also chordal graphs and perfect graphs are expandable.

Chordal graphs are expandable because their excluded subgraphs are all cycles of size 4 or more, which have no true twins (though cycles of length 4 have two pairs of false twins). Similarly the excluded subgraphs for perfect graphs are odd cycles of length at least 5 and their complements. These odd cycles have no true or false twins, and it is easy to see that a graph $G$ has true twins $u,v$ if and only if $u,v$  are false twins in the complement $\overline{G}$. Thus there cannot be any true twins in the complements of the odd cycles.

Planar graphs are not expandable by Lemma~\ref{lem:2-module}:
The complete graph $K_5$ is a minimal excluded subgraph, and each
pair of vertices are true twins.

Triangle-free graphs are also not expandable, because their single excluded subgraph is the triangle, and  any pair of vertices of a triangle are true twins.
Hence bipartite graphs are not expandable, since the triangle is a minimal excluded subgraph.

Line graphs are also not expandable, because each of the graphs $\mathrm{G}_i$ ($i\in\{2,3,4,5,6,7\}$) in Beineke's list~\cite{beineke} of excluded subgraphs contains a pair of true twins. Thus the transformation $(G,w)\mapsto G_w$ was not available to Jerrum and Sinclair~\cite{JS} in their work on weighted matchings, and so they
had to re-prove their results for the weighted case.

\subsubsection{Application to sampling}
The use of the transformation  $(G,w)\mapsto G_w$  here is as follows.
If all the weights $w(v)$ are bounded by a small polynomial, say $w(v)=O(n^r)$ for all $v\in [n]$,
then we can use the transformation as a way to extend an algorithm for the unweighted case to an algorithm for the
weighted case. We run the unweighted algorithm on a graph of size $O(n^{r+1})$,
so the running time remains polynomial.

However, to carry out the Glauber dynamics, as we do in Section~\ref{ss:glauber}, it is unnecessary to construct $G_w$ explicitly. We can simulate the chain on $G_w$ using only $G$. Let $w_+(V)=\sum_{v\in V}w(v)$. Then, at each step of the dynamics, we choose vertex $v\in V$ with probability $w(v)/w_+(V)$. This corresponds to choosing a vertex of $G_w$ uniformly, as in Section~\ref{ss:glauber}. Then, if $Z$ is the current independent set in $G$,
\begin{itemize}[itemsep=0pt,topsep=0pt]
\item If $v\in Z$ then $Z'\gets Z\sm\{v\}$ with probability $1/(2w(v))$. This corresponds to selecting the unique vertex in the independent set in the $w(v)$-clique in $G_w$ and deleting it with probability $\nicefrac12$.
\item If $v\notin Z$ and $Z\cup\{v\}\in\cI(G)$ then  $Z'\gets Z\cup\{v\}$ with probability $\nicefrac12$. This corresponds to selecting a random vertex in the currently unoccupied $w(v)$-clique in $G_w$ and inserting it with probability $\nicefrac12$.
\item Otherwise $Z'\gets Z$.
\end{itemize}
The analysis is identical to that leading to Theorem~\ref{thm:main-Cp}, except that $n$ must now be replaced by $w_+(V)$, the number of vertices in the expanded graph.

\subsection{Graphs with large complete bipartite subgraphs}\label{s:largebip}

In Lemma~\ref{lem:notes} we observed that if a graph
$G$ contains $K_{d,d}$ as an induced subgraph then its pathwidth is at
least $d$.  Thus our argument does not guarantee rapid mixing for
any graph $G$ which contains a large induced complete bipartite subgraph.
In this section we show that the absence of large induced complete bipartite
subgraphs appears to be a necessary condition for rapid mixing.

Suppose that the graph $G$ consists of $k$ disjoint induced copies of $K_{d,d}$. So
$n=2kd$.  The state space $\cI_G$ of independent sets in $G$ has
$|\cI_G|=2^{k+d}$.  The mixing time for the Glauber dynamics on $G$ is clearly
at least $k$ times the mixing time on $K_{d,d}$.

Now consider the $K_{d,d}$ with vertex bipartition $L\cup R$.
The state space $\cI_K$ of independent sets of $K_{d,d}$ comprises two
sets $\cI_L=\{I\in \cI_k \, :\,  I\cap R=\es\}$
and  $\cI_R=\{I\in \cI_k \, :\,  I\cap L=\es\}$.
Now $\cI_L\cap\cI_R = \{\es\}$, $|\cI_L|=|\cI_R|=2^d$, and
$|\cI_L\cap\cI_R|=1$. It follows that the conductance of the Glauber dynamics
on $G$ is $O(2^{-d})$, and so $\tau_\es(\eps)=\Omega(2^d\log(1/\eps))$.
(See~\cite{jerrumbook} for the definition of conductance.)
If $d=\omega\log_2 n$, where $\omega\to\infty$ as $n\to\infty$,  then
$\tau_\es(\eps)=\Omega(n^\omega\log(1/\eps))$, so the Glauber dynamics is not an FPAUS.

Note that, if $d=O(\log n)$ then the Glauber dynamics has quasipolynomial mixing
time, from Theorem~\ref{thm:main-Cp}, whereas our lower bound remains
polynomial. Our techniques are insufficient to distinguish between
polynomial and quasipolynomial mixing times.

These graphs are not connected, but the analysis is little changed if the
$K_{d,d}$ components are loosely connected, for example sequentially by a single edge.

Theorem~\ref{thm:main-Cp} shows that the Glauber dynamics for independent sets
is rapidly mixing for any graph $G$ in the class $\cC_p$, where $p$ is a fixed positive integer.
However, it is not clear a priori which graphs belong to $\cC_p$,  and the
complexity of recognising membership in the class $\cC_p$ is currently unknown,
though Mann and Mathieson~\cite{MM} report that this problem is W[1]-hard
in the worst case.

Therefore, in the remainder of this paper we consider
(hereditary) classes of graphs which are determined by small excluded subgraphs. These
classes clearly have polynomial time recognition, though we will not be concerned with
the efficiency of this. Note that, in view of Section~\ref{s:largebip}, we must always
explicitly exclude large complete bipartite subgraphs, where this is not already implied
by the other excluded subgraphs.

The three classes we will consider are nested. The third includes the second, which
includes the first. However, we will obtain better bounds for pathwidth in the smaller classes,
and hence better mixing time bounds in Theorem~\ref{thm:main-Cp}.  Therefore we consider them
separately.  The first of these classes, claw-free graphs,
was considered by Matthews~\cite{JM} and forms the motivation for this work.
Our results on claw-free graphs are found in Section~\ref{s:claw-free},
while the other two classes are described in Section~\ref{s:classes}.

%
%
%

\section{Claw-free graphs}\label{s:claw-free}

Claw-free graphs exclude the following induced subgraph, the \emph{claw}.
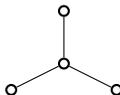
\begin{center}
  \begin{tikzpicture}[scale=0.35]
    \node[w] (v4) at (2,0) {};
    \node[w] (vv) at (4,1) {};
    \node[w] (v0) at (4,3) {};
    \node[w] (v2) at (6,0) {};
    \draw (v4)--(vv)--(v0)  (vv)--(v2);
  \end{tikzpicture}
\end{center}
Claw-free graphs are important because they are a simply characterised superclass of \emph{line graphs}~\cite{beineke},
and independent sets in line graphs are \emph{matchings}.

For claw-free graphs, the key observation is as follows.

\begin{lemma}\label{lem:claw}
Let $G$ be a claw-free graph with independent sets $X, Y\in\cI(G)$.
Then $G[X\sd Y]$ is a disjoint union of paths and cycles.
\end{lemma}

\begin{proof}
We know that $G[X\sd Y]$ is an induced bipartite subgraph of $G$.
Since $G$ is claw-free, any three neighbours of a given vertex must span at
least one triangle (3-cycle).  But this is impossible, since $G[X\sd Y]$
is bipartite. Hence every vertex in $G[X\sd Y]$
has degree at most 2, completing the proof.
\end{proof}

\begin{lemma}\label{lem:C_2}
  Claw-free graphs are a proper subclass of $\cC_2$.
\end{lemma}

\begin{proof}
  From Lemma~\ref{lem:claw}, the bipartite subgraphs of $G$ are path and cycles.
  From \eqref{eq:pw}, these have pathwidth at most 2. So $\bpw(G)\leq 2$.
  On the other hand, there are many bipartite graphs with pathwidth 2 which are not claw-free.
  For example $K_{2,b}$ has pathwidth~2, from \eqref{eq:pw}, but contains
  claws if $b\geq 3$.
\end{proof}
Since $G[X\sd Y]$ is a union of paths and even cycles, the
Jerrum--Sinclair~\cite{JS} canonical paths for matchings can be adapted, and
the Markov chain has polynomial mixing time. This was the idea employed
by Matthews~\cite{JM}. Theorem~\ref{thm:main-Cp} generalises his result to $\cC_p$
for any fixed integer $p\geq 3$.

However, claw-free graphs have more structure than an arbitrary graph in $\cC_2$,
and this structure was exploited for matchings in~\cite{JS}. Note that when $G$
is claw-free, we can compute the
size $\alpha(G)$ of the largest independent set in $G$  in polynomial time~\cite{Minty},
just as we can compute the size of the largest matching~\cite{edmonds}.

In Section~\ref{ss:sample-size} we strengthen and extend the results of~\cite{JS} to all claw-free graphs.
Our main extension is to show how to sample almost uniformly from
$\cI_k(G)$ more directly for arbitrary $k$. Jerrum and Sinclair's procedure is to estimate $|\cI_i(G)|$
successively for $i=1,2,\ldots,k$, which is extremely cumbersome. However, we should add
that their main objective is to estimate $|\cI_\alpha(G)|$, rather than to sample.

\subsection{Sampling independent sets of a given size in claw-free graphs}\label{ss:sample-size}

Hamidoune~\cite{hamidoune} proved that in a claw-free graph $G$, the numbers $N_i$ of
independent sets of size $i$ in $G$ forms a log-concave sequence.
Chudnovsky and Seymour~\cite{CS} strengthened this, by showing that all the roots of $P_G(x)=\sum_{i=0}^\alpha N_ix^i$ are real. If $\lambda>0$, let $M_i=\lambda^iN_i$  for $i=0,1,\ldots, \alpha$,
so $P_G(\lambda)=\sum_{i=0}^{\alpha}M_i$. Clearly the polynomial $P_G(\lambda x)=\sum_{i=0}^\alpha M_ix^i$ also has real roots, as does the polynomial $x^\alpha P_G(1/x)=\sum_{i=0}^\alpha N_{\alpha-i}\, x^i$.  Thus we can equivalently use the sequences $\{M_i\}$, $\{M_{\alpha-i}\}$.

Since $P_G(\lambda x)$ has only real roots, it follows
(see~\cite[Lemma~7.1.1]{Branden})
 that $\{ M_i/\binom{\alpha}{i}\} $ is a log-concave sequence, for a fixed value of
$\lambda$.
From this we have, for $1\leq i\leq \alpha-1$,
\begin{equation}
 \frac{M_{i-1}}{M_i}\ \leq\ \frac{i(\alpha-i)}{(i+1)(\alpha-i+1)}\frac{M_i}{M_{i+1}}\ \leq\ \frac{i}{i+1}\frac{M_i}{M_{i+1}}.\label{eq:real}
\end{equation}
We use this to strengthen an inequality deduced in~\cite{JS} for log-concave functions.
\begin{lemma}\label{logconcave}
For any $m\in [n]$ and $k\in [m]$,
\[ \frac{M_{m-k}}{M_m}\, \leq \, e^{-(k-1)^2/2m}\left(\frac{M_{m-1}}{M_m}\right)^k\,.\]
\end{lemma}

\begin{proof}
We first prove by induction on $m$ and $k$ that
\begin{equation}
\label{ind-hyp}
 \frac{M_{m-k}}{M_m} \leq \frac{(m)_{k}}{m^{k}}\, \left(\frac{M_{m-1}}{M_m}\right)^k
\end{equation}
for all $m\in [n]$ and all $k\in [m]$.

If $k=1$ then (\ref{ind-hyp}) holds with equality.
So the base cases for the induction are $(m,1)$ for all $m\in [n]$.
For the inductive step, we assume that the result holds for $(m-1,k)$ for some
$k\in [m-1]$,  and wish to conclude that the result holds for $(m,k+1)$.
Now
\[
 \frac{M_{m-(k+1)}}{M_{m}}\ =\ \frac{M_{(m-1)-k}}{M_{m-1}}\cdot \frac{M_{m-1}}{M_{m}}
 \leq\ \frac{(m-1)_{k}}{(m-1)^{k}}\, \left(\frac{M_{m-2}}{M_{m-1}}\right)^k\,
  \frac{M_{m-1}}{M_{m}}\,,
\]
using the inductive hypothesis for $(m-1,k)$.
Now apply (\ref{eq:real}) with $i=m-1$ to obtain
\begin{align*}
 \frac{M_{m-(k+1)}}{M_{m}}\ &
   \leq\ \frac{(m-1)_{k}}{(m-1)^{k}}\, \left(\frac{m-1}{m}\right)^k\,
   \left(\frac{M_{m-1}}{M_{m}}\right)^{k+1} \\
   &=\ \frac{(m-1)_{k}}{m^k}\,   \left(\frac{M_{m-1}}{M_{m}}\right)^{k+1}
   \ =\ \frac{(m)_{k+1}}{m^{k+1}}\left(\frac{M_{m-1}}{M_{m}}\right)^{k+1}.
\end{align*}
This shows that (\ref{ind-hyp}) holds for $(m,k+1)$, completing the inductive step.
Hence (\ref{ind-hyp}) holds for all $m\in [n]$ and $k\in [m]$, and the lemma
follows since
\begin{align*}
\frac{(m)_{k}}{m^{k}}\ =\ \prod_{j=0}^{m-1} \left(1-\frac{j}{m}\right)
&\leq\ \exp\bigg( -\dfrac{1}{m} \sum_{j=0}^{m-1} j\bigg)\,\qquad\text{using }1-x\leq e^{-x},\\
&=\  \exp\bigg(-\frac{k(k-1)}{2m}\bigg)\ \leq\ e^{-(k-1)^2/2m}\,.
\end{align*}
\end{proof}

\bigskip

Now suppose that $M_m=\max_i M_i$. Then Lemma~\ref{logconcave} implies that
$M_{m-k}\leq e^{-(k-1)^2/2m}\, M_m$
for all $k\in[m]$.  Since the polynomial $\sum_{i=0}^\alpha M_{\alpha-i}x^i$ only
has real roots, applying Lemma~\ref{logconcave} to this polynomial gives
$M_{m+k}\leq e^{-(k-1)^2/2m}\, M_m$  for all $k\in[\alpha-m]$.
Therefore
\begin{align*}
 P_G(\lambda)=\sum_{i=0}^\alpha M_i
  &\leq M_m\left(1 + \sum_{k=1}^m e^{-(k-1)^2/(2m)} + \sum_{k=1}^{\alpha-m} e^{-(k-1)^2/(2m)}\right)\\
 &\leq  2 M_m\int_{0}^{\infty}\!e^{-x^2/2m}\,\mathrm{d}x = M_m\sqrt{2\pi m}\, .
\end{align*}
Thus, if $Z$ is  a random independent set drawn from the stationary distribution $\pi$ of the Glauber dynamics, then
\begin{equation}\label{p_m}
\Pr(|Z|=m)\,=\,\frac{M_m}{P_G(\lambda)}\,\geq\,\frac{1}{\sqrt{2\pi m}}\,.
\end{equation}
By choosing $\lambda$ appropriately, we can take $m$ to be any value  $i\in[\alpha]$.

\medskip

Define $\lambda_i=N_{i-1}/N_i$ for all $i\in [\alpha]$.
Now (\ref{eq:real}) implies that
\begin{equation}
\label{monotonic-lambda}
\lambda_1 < \lambda_2 < \cdots < \lambda_{\alpha-1} < \lambda_\alpha.
\end{equation}
For $i$ to be the maximiser we require
\[  N_{i-1}\lambda^{i-1}\leq N_i\lambda^i\quad \text{ and } \quad
   N_{i+1}\lambda^{i+1}\leq N_i\lambda^i,\]
that is, $\lambda_i\leq\lambda\leq \lambda_{i+1}$.
A suitable value of $\lambda$ exists, by (\ref{monotonic-lambda}).
With this value of $\lambda$, we need $O(\sqrt{m})$ repetitions of the chain to obtain one
sample from $\cI_m(G)$. We explain below how to determine an appropriate value for $\lambda$.

For $m=\alpha$, we need $\lambda \geq \lambda_\alpha=N_{\alpha-1}/N_\alpha$. Following~\cite{JS}, we may take $\lambda=2\lambda_\alpha$.
Then $M_{\alpha-1}/M_\alpha\ =\nicefrac12$, and hence from Lemma~\ref{logconcave}, $M_{\alpha-k}/M_\alpha < 1/2^k$ for $k\in[\alpha]$. Hence $\sum_{i=0}^\alpha M_i < 2M_{\alpha}$,
and so, for $Z$ in stationarity, $\Pr(|Z|=\alpha)>\nicefrac12$.
Thus we need only $O(1)$ repetitions of the chain to get
one almost-uniform sample from $\cI_\alpha(G)$. Of course, the Markov chain only gives approximate samples,
but the small distance from stationarity does not alter this conclusion.

From Theorem~\ref{thm:main-Cp}, the Glauber dynamics will have polynomial mixing time in $n$ if $\lambda$ is polynomial in $n$. Thus we will require $2\lambda_\alpha\leq n^q$, for some constant $q$, in the family of graphs we consider. Jerrum and Sinclair~\cite{JS} called such a family of graphs \emph{$n^q$-amenable}. Clearly not all claw-free graphs are $n^q$-amenable, for any constant $q$, since there exist line graphs which are not $n^q$-amenable~\cite{JS}. Also, to apply the algorithm, we need an explicit polynomial bound for~$\lambda_\alpha$. We consider below how this can be obtained.

We have determined the best interval of $\lambda$ for sampling independent sets of size $m$, but we need to find this interval.
To this end we must consider how $\Pr(|Z|=m)$ varies with $\lambda$. We will denote this by $p_m(\lambda)$, or by $p_m$ if $\lambda$ is fixed.

\begin{lemma}\label{unimodal} Fix $m\in [\alpha]$.
Then $p_m(\lambda)=\Pr(|Z|=m)$ is a unimodal function of $\lambda > 0$.
\end{lemma}

\begin{proof}
Since all roots of
$P_G(\lambda)$ are real and negative, we can write $P_G(\lambda) = \prod_{i=1}^{n}(\sigma_i+\lambda)$ for positive constants $\sigma_1,\ldots, \sigma_n$.
Since $\lambda > 0$ we can write $\lambda=e^x$.  Then
\[ p_m(\lambda)\,=\,\frac{N_m\lambda^m}{P_G(\lambda)}\,=\,\frac{N_m\lambda^m}{\prod_{i=1}^{n}(\sigma_i+\lambda)} \,=\,\frac{N_m e^{mx}}{\prod_{i=1}^{n}(\sigma_i+e^x)}\,=\,f(x)\,, \]
say. Thus
\begin{align*}
  L(x)\,&=\,\ln f(x)\,=\,-\sum_{i=1}^{n}\ln(\sigma_i+e^x)+mx+\ln N_m\,, \\
  L'(x)\,&=\,-\sum_{i=1}^{n}\frac{e^x}{(\sigma_i+e^x)}+m\,=\,\sum_{i=1}^{n}\frac{\sigma_i}{(\sigma_i+e^x)}+m-n\,,\\
  L''(x)\,&=\,-\sum_{i=1}^{n}\frac{\sigma_ie^x}{(\sigma_i+e^x)^2}\,<\,0.
\end{align*}
Thus $L(x)$ is concave, and hence a unimodal function of $x$. Since $\lambda=e^x$ is an increasing function of $x$, this implies $p_m(\lambda)$ is also unimodal as a function of $\lambda > 0$.
\end{proof}

As noted above (\ref{eq:real}), the sequence $\{ M_i/\binom{\alpha}{i} \}$ is log-concave,
for a fixed value of $\lambda >0$.  It follows from this (see~\cite[Lemma~7.1.1]{Branden})
that for a fixed value of $\lambda$,
the sequences $\{ M_i\}$ and $\{ p_i(\lambda)\}$ are both log-concave.

We can now return to sampling from $\cI_m(G)$ for any $m\leq\alpha$.
We will determine a suitable $\lambda$ for sampling from the Gibbs distribution by using bisection
to approximately maximise the unimodal function $p_m(\lambda)$.
The algorithm will work well whenever $\alpha \gg \log^2 n$.

\subsubsection{The bisection process}\label{ss:bisection}

The bisection process will use information obtained from the Glauber dynamics to
update a left marker $\kappa_0$ and a right marker $\kappa_1$
which satisfy
\[ 0 < \kappa_0<\lambda_m<\kappa_1<2\lambda_m<n^q.\]
We will describe below how to select the initial values of $\kappa_0$, $\kappa_1$.
At each bisection step, the midpoint of the interval $[\kappa_0,\kappa_1]$ will
become the new left marker (respectively, new right marker), if the current value of
$\lambda$ is too small (respectively, too large).

We would ideally wish to find a point in the interval [$\lambda_m,\lambda_{m+1}$],
as then
$p_m(\lambda)\geq 1/\sqrt{2\pi m} \approx 0.399/\sqrt{m}$. However, since we can only approximate $p_m(\lambda)$,
we merely
seek a $\lambda$ such that $p_m(\lambda)\geq 0.156/\sqrt{\alpha}$, with moderately
high probability.
By Lemma~\ref{unimodal} and \eqref{p_m}, given a positive constant $c < 1/\sqrt{2\pi}$
and any $m\in [\alpha]$,
the values of $\lambda$ such that $p_m(\lambda)>c/\sqrt{\alpha}$ form a non-empty interval
$\Lambda$, and $[\lambda_m,\lambda_{m+1}]\subseteq\Lambda$ by \eqref{p_m}.

To estimate how many steps of bisection might be required, we need to show that the
interval $\Lambda$ cannot be too short. In the worst case,
$\Lambda = [\lambda_m,\lambda_{m+1}]$. From \eqref{eq:real} with $i=m$, we have
\begin{equation}\label{eq:interval}
  \lambda_{m+1}-\lambda_m\ \geq\ \lambda_m/m\,.
\end{equation}
Similarly to~\cite{JS} and above, we will require that $\lambda_m\leq n^q$, for some constant $q$.
Note that $ \lambda_m\leq \lambda_\alpha$ by (\ref{monotonic-lambda}), so this is implied by $n^q$-amenability, though the converse may not hold.
Then, since the initial interval will have width less than $2\lambda_m$, we can locate a
point in $[\lambda_m,\lambda_{m+1}]$ in at most
$\log_2(2m\lambda_m)\leq \log_2(2m\kappa^*)$
iterations of bisection, where $\kappa^*$ is the initial value of $\kappa_1$.
Since $2\kappa^*\leq n^q$ for some constant $q$, this is at most $(q+1)\log_2 n$
 bisection steps.

We now describe how to carry out a step of the bisection.

\medskip

{\bf Bisection:}\  Let $\lambda = (\kappa_0+\kappa_1)/2$.\\[0.5ex]
{\bf Burn-in:}\ Run the Glauber dynamics
for $\tau_\es(1/n)$ steps, with this $\lambda$, from initial state $\es$.\\[0.5ex]
{\bf Estimation:}\  Run the Glauber dynamics for a further
$N=\lceil 9\alpha\, (1-\beta_1)^{-1}\rceil$ steps (still with this $\lambda$), obtaining a sample
of $N$ independent sets $X_1,X_2,\ldots,X_N$.

Let $\zeta_{j,i}$ be the indicator variable which is~1 if $X_j\in\cI_i$, and~0 otherwise.
Thus $\eta_i=\sum_{j=1}^N\zeta_{j,i}$ is the number of occurrences of an element of
$\cI_i(G)$ during these $N$ steps.
Let $\xi_i = \eta_i/N$ for $i\in [\alpha]$, and note that
$\xi_i$ is our estimate of $p_i(\lambda)$.
\begin{enumerate}[itemsep=0pt,label=(\arabic*)]
\item Let $\Xi=\{i \in [\alpha]\, :\, \xi_i\geq 0.328/\sqrt{\alpha}\}$,\quad $\Xi'=\{i\in [\alpha]\, :\, \xi_i\geq 0.207/\sqrt{\alpha}\}$.
Then $\Xi\subseteq \Xi'$, and note that $\Xi\neq \emptyset$ by definition of
$N$.
\item If $m\in \Xi'$ then stop.  We will conclude that $p_m(\lambda)\geq 0.156/\sqrt{\alpha}$
and that the current value of $\lambda$ is suitable for sampling from $\mathcal{I}_m(G)$.
\item Otherwise, choose $k\in\Xi$ uniformly at random.
\item If $k>m$ then we conclude that $\lambda_m<\lambda$ and
move left in the bisection
by setting the right marker $\kappa_1$ to $\lambda$.
\item If $k<m$ then we conclude that $\lambda_m>\lambda$ and
move right in the bisection
by setting the left marker $\kappa_0$ to $\lambda$.
\item If we have performed more than $\log_2(2m\kappa^*)$ iterations then stop:
the bisection process has failed.
\item Otherwise, begin the next bisection step with the new values of $\kappa_0$, $\kappa_1$.
\end{enumerate}
Note that there is a small probability that the conclusions that we make
in steps (2), (4) and (5) are wrong.

\bigskip

We now analyse this bisection process.
Suppose that $\lambda\in[\lambda_s,\lambda_{s+1}]$ during the current bisection step.
Also assume that we do not terminate during line (2), so $m\not\in \Xi'$.
For ease of notation, write $p_i=p_i(\lambda)=M_i/P_G(\lambda)$ for $i\in[\alpha]$, using
the current value of $\lambda$.
Then $p_s\geq 1/\sqrt{2\pi\alpha}$ from~\eqref{p_m}.
It follows from~\cite[(4.7)]{Aldous}  that
$\E[\xi_i]=p_i$ and
\[\var(\xi_i)=\E[(\xi_i-p_i)^2]< 2  \left(1+\frac{1}{n}\right) \left(1+\frac{e^{-9\alpha}}{9\alpha}\right)\frac{p_i}{9\alpha} < \frac{3p_i}{9\alpha}=\frac{p_i}{3\alpha}\,\]
for all $i\in[\alpha]$.
Thus, using the Chebyshev inequality,
\[
 \Pr\big(|\xi_i-p_i|>\sqrt{p_i}/(9\alpha^{1/4})\big) < \frac{81\sqrt{\alpha}}{p_i} \cdot
\frac{p_i}{3\alpha}=\frac{27}{\sqrt{\alpha}}
\]
for all $i\in [\alpha]$. Letting $p_i=c_i/\sqrt{\alpha}$ for $i\in [\alpha]$,
we can rewrite the above inequality as
\begin{equation}\label{eq:concentration}
 \Pr\Big(\xi_i\notin (c_i\pm\sqrt{c_i}/9)/\sqrt{\alpha}\Big) < \frac{27}{\sqrt{\alpha}}\,.
\end{equation}
In the remainder of this section, we write ``with high probability'' to mean ``with probability at least $1-27/\sqrt{\alpha}$'', unless otherwise stated.

In particular, \eqref{eq:concentration} implies that
$\xi_s > (c_s-\sqrt{c_s}/9)/\sqrt{\alpha}$ with high probability.
Since $c_s\geq 1/\sqrt{2\pi}>0.39$, it follows from (\ref{eq:concentration}) that
\begin{equation}
\label{xis-bound}
\xi_s>0.32/\sqrt{\alpha}
\end{equation}
with high probability.
Thus, we have $s\in\Xi$ with high probability.

If $m\in\Xi$, then clearly $m\in\Xi'$, and we would have terminated in line~(2).
So $m\notin\Xi$.

Now the lower bound on $\xi_i$ from (\ref{eq:concentration}) can be rearranged to give
\[ c_i^2 - (2b_i + \dfrac{1}{81})c_i + b_i^2 \leq 0,\]
where $\xi_i = b_i/\sqrt{\alpha}$.  Solving for $c_i$ tells us that with high probability,
\begin{equation}
\label{inverse-concentration}
 c_i \in b_i + \dfrac{1}{162} \pm  \sqrt{\frac{b_i}{81}  + \frac{1}{162^2}}.
\end{equation}
Since $k$ is chosen randomly from $\Xi$, we must have $\xi_k > 0.328/\sqrt{\alpha}$
and hence, by (\ref{inverse-concentration}) we conclude that
\begin{equation}
\label{ck-bound}
c_k>0.27
\end{equation}
with high probability.

Now $m\neq k$ since we did not terminate in line~(2).  Suppose that $m > k$.
(The case $m< k$ is symmetrical and hence will be omitted.)

Recall that $\{ p_i(\lambda)\}$ is a log-concave sequence, for a fixed $\lambda$, and hence is unimodal.
Therefore, if $k<m<s$ then $c_m\geq \min\{c_k,c_s\}$, which (assuming
(\ref{ck-bound}) holds) implies that
$c_m>0.27$.   Applying (\ref{eq:concentration}) shows that in this case we have
\begin{equation}
\label{xim-bound}
\xi_m>0.21/\sqrt{\alpha}
\end{equation}
with high probability.  Now (\ref{xim-bound}) implies that $m\in\Xi'$, which contradicts
the fact that the process did not terminate in line~(2).
Thus, we can assume that $m\geq s$, but again $m=s$ is impossible, as we did not terminate in
line~(2).  Therefore $m\geq s+1$, and so by (\ref{monotonic-lambda}) we have
$\lambda_m\geq \lambda_{s+1}>\lambda$. Hence the bisection in line~(5) is correct,
since the current value of $\lambda$ is indeed too small.

Finally, suppose that we do terminate in line~(2) in the current bisection step.
This happens only when $\xi_m \geq 0.207/\sqrt{\alpha}$,
and by (\ref{inverse-concentration}) we can conclude that
\begin{equation}
\label{cm-bound}
 c_m > 0.162
\end{equation}
with high probability.

We can make an error at any step of the bisection by terminating early, or by failing to
terminate. This will happen only if  one of $\xi_s$, $c_k$, $\xi_m$ (or $c_m$, in the
very last step) does not satisfy its
assumed bounds. That is, a non-final step will give an error only if
one of (\ref{xis-bound}), (\ref{ck-bound}), (\ref{xim-bound}) fails,
while the final step fails only if (\ref{cm-bound}) fails.
Thus the failure probability at each step is at most $3\times 27/\sqrt{\alpha}=81/\sqrt{\alpha}$. If $\kappa^*<n^q$ for some $q=O(1)$, and we bisect for $\log_2 (2m\kappa^*)\leq (q+1)\log_2 n$ steps, then the overall probability of an error during the bisection process is
at most
\[ 81\big((q+1)\log_2 n + 1\big)/\sqrt{\alpha}\]
of making any error during the bisection. This is small for large enough $n$,
provided that $\alpha\gg \log^2n$.

If $\alpha = O(\log^2 n)$, adapting the ``incremental'' approach of \cite{JS} would be
little worse than the bisection method, so we will not consider this case further.

We can now use this value of $\lambda$ to sample from $\cI_m(G)$. If, during this sampling,
we detect that $p_m < 0.162/\sqrt{\alpha}$, so the bisection terminated incorrectly,
we repeat the bisection process.

Clearly, very few repetitions are needed until the overall probability of failure is as small as desired.
Since $\alpha \leq n$ and $q=O(1)$, it follows that the total time required for this bisection process is only
$O(\log n\cdot \tau_\es(\eps))$.

\subsubsection{Finding the initial left and right marker}

We now show how to find initial values for $\kappa_0,\kappa_1$ using a standard ``doubling''
device. We use the approach of the bisection.  Initially, let $\lambda = 2e/n$,
and iterate as follows.

With the current value of $\lambda$, we perform a step of the bisection process above.
If we conclude that $\lambda_m<\lambda$, then we stop, setting
the values of the initial left marker $\kappa_0 = \lambda/2$ and the
initial right marker $\kappa^*=\kappa_1 = \lambda$.
If this occurs at the first step, with $\lambda =2e/n$, we are in the situation
where $\lambda_m<e/n$, discussed in Section~\ref{ss:notation}, so we proceed
deterministically, as outlined there.

Otherwise, we double $\lambda$ and repeat this with the new value of~$\lambda$.
After $i$ iterations of this doubling we have $\lambda = 2^{i+1}\, e/n$.
Therefore, when $i=\lceil \log_2(n\lambda_m/(2e))\rceil$
we have $\lambda/2 < \lambda_m < \lambda$ with high probability.

Of course, there is a small probability that we stop too early;
leading to $\kappa^* = \kappa_1<\lambda_m$.
If this occurs, we would detect it when the subsequent bisection
process fails. Then we would simply repeat the doubling process.

\subsubsection{Running time}

The initial phase (to find the initial left and right marker) requires only
$O(\log_2(n\kappa^*))$ iterations.  So if $\lambda_m=O(n^q)$ for $q=O(1)$ then only $O(\log n)$ iterations are required to find the initial left and right marker.

For each bisection step, the time required  for the ``burn in'' is $\tau_\es(\eps)$
and then $N$ further Markov chain steps are performed, to find the samples
$X_1,\ldots, X_N$.
The time required for this further
sampling will be $O(\tau_\es(\eps))$ if $n\lambda > e$, 
by (\ref{both-eigvals}), assuming that $\beta_{\max} = \beta_1$.

We saw earlier that $O(\ln n)$ bisection steps are required to find a suitable
$\lambda$, and that the failure probability is $o(1)$ if $\alpha\gg \log^2 n$.

Therefore, the total running time of the whole bisection process is then $O(\tau_\es(\eps)\cdot\log n)$, and throughout the bisection process the  Markov chain is always run using values of $\lambda$ which satisfy $\lambda\leq 2\lambda_m$. Thus the mixing time bound is at most $2^{p+1}=8$ times larger than that with $\lambda=\lambda_m$, from Theorem~\ref{thm:main-Cp}. This is a modest constant factor, so our bisection procedure compares very favourably with the running
time of $\Omega( \tau_\es(\eps)\cdot\alpha^2)$ obtained using the bisection method given
by Jerrum and Sinclair in~\cite{JS}.
In their bisection process, the Markov chain is run only with values $\lambda\leq\lambda_m$.

A final remark: In most practical situations we do not know the relaxation time or the mixing time exactly;
rather, we know an upper bound $R$ on the relaxation time and $T(\eps)$ on the mixing time.
If the bound $T(1/n)$ is obtained using (\ref{both-eigvals}) then
\[ T(\eps) =R \big(\alpha \ln(n\lambda) + \ln(n) + 1\big),\]
assuming that $\beta_{\max} = \beta_1$.
In this case we should take $N = \lceil 9\alpha R\rceil$ samples in each bisection step,
and it is still true that this additional sampling takes at most a constant times longer than the burn-in,
assuming that $\lambda > e/n$.

\subsubsection{Vertex weights}\label{sec:realroots}
Using the results of Section~\ref{ss:vertexweights}, we can extend the above analysis to the
case where all weights $w(v)$ are polynomially bounded in $n$. Then we can use the bisection method above to estimate $W_k(G)$, provided that $W_{\alpha-1}/W_\alpha$ is also polynomially bounded. This will complete our generalisation of all the results of \cite{JS} from line graphs to claw-free graphs.

The independence polynomial $P_G(\lambda)$ of $G$ can be generalised to the weighted case by defining
\[
P^w_G(\lambda)=\sum_{I\in\cI(G)}w(I)\lambda^{|I|}=\sum_{k=0}^{\alpha(G)} W_k(G)\, \lambda^k.
\]
For claw-free graphs, our algorithm for approximating $N_k(G)$ is based on Chudnovsky and Seymour's result~\cite{CS} that $P_G(\lambda)$ has only negative real roots in $\lambda$. This implies the strong form of log-concavity for the coefficients $N_k(G)$ that we used in Section~\ref{ss:bisection}. We will now show that the transformation of Section~\ref{ss:vertexweights} allows us to extend the result of~\cite{CS} to $P^w_G(\lambda)$.
\begin{lemma}
Suppose that $P_G(\lambda)$ has only negative real roots for all $G$ in an expandable class $\cC$.
Then, for any $G\in\cC$ and real weights $w(v)\geq 0$, the polynomial $P^w_G(\lambda)$ has only negative real roots.
\end{lemma}

\begin{proof}

First, suppose that all $w(v)$ are integers. Then the conclusion follows
from the equivalence of $G$ and $G_w$.

Next, suppose that all $w(v)$ are rational. Let $q$ be the lowest common divisor of the weights $w(v)$,
so $w(v)=w'(v)/q$ for some integer $w'(v)$, for all $v\in V$. Then $W_k=W'_k/q^k$, where $W'_k$ is calculated
using the integer weights $w'(v)$. Thus
\begin{equation*}
P^w_G(\lambda)=\sum_{k=0}^{\alpha} W_k\, \lambda^k= \sum_{k=0}^{\alpha} (W'_k/q^k)\, \lambda^k = P^{w'}_G(\lambda/q)\,.
\end{equation*}
Hence, since $P^{w'}_G(\lambda)$ has only negative real roots, so does $P^w_G(\lambda)$.

Finally, suppose that at least one $w(v)$ is irrational. For each $v\in V$, let
$\{w_i(v)\}$ ($i=1,2,\ldots$) be a sequence of rational numbers which converge to $w(v)$.
Then, from~\cite{HaMa}, the negative real roots of $P^{w_i}_G(\lambda)$ converge to the roots of $P^w_G(\lambda)$,
which must therefore be real and nonpositive.  However, $P^w_G(0) = N_0(G) = 1$, so 0 is not a root of
$P^w_G(\lambda)$ and hence all roots of $P^w_G(\lambda)$ are negative, completing the proof.
\end{proof}

\section{Other recognisable subclasses of $\cC_p$} \label{s:classes}

We now consider two more classes of graphs which we will show to have bounded bipartite pathwidth.

\subsection{Graphs with no fork or complete bipartite subgraph}\label{ss:fork}

Fork-free graphs exclude the following induced subgraph, the fork:
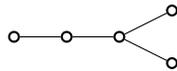
\begin{center}
  \begin{tikzpicture}[scale=0.35]
    \node[w] (v5) at (0,1) {};
    \node[w] (v4) at (2,1) {};
    \node[w] (vv) at (4,1) {};
    \node[w] (v0) at (6,2) {};
    \node[w] (v2) at (6,0) {};
    \draw (v5)--(v4)--(vv)--(v0)  (vv)--(v2);
  \end{tikzpicture}
\end{center}
We characterise fork-free bipartite graphs. Recall that two vertices $u$ and $v$ are
\emph{false twins} if $\Nb(u)=\Nb(v)$.
In Figure \ref{fig:cbff}, vertices to which false twins can be added are
filled (black). Hence each graph containing a black vertex represents
an infinite family of augmented graphs. For instance, $P_2^*$ represents all complete
bipartite graphs, $P_4^*$ represents graphs obtained from $K_{a,b}$ by
removing one edge; removing two edges leads to an augmented domino.

\begin{figure}[htbp]
  \hspace*{\fill}
  \begin{tikzpicture}[xscale=0.4, yscale=0.5] 
    \foreach \x in {1,3,5,7,9} \node[w] (\x) at (\x,0) {};
    \foreach \x in {2,4,6,8}   \node[w] (\x) at (\x,1) {};
    \foreach[count=\n] \x in {2,3,...,9} \draw (\n)--(\x);
  \end{tikzpicture}
  \hfill
  \begin{tikzpicture}[scale=0.6] 
    \foreach \a in {1,2,...,9} \node[w] (\a) at (22.5+45*\a:1) {};
    \foreach[count=\n] \a in {2,3,...,9} \draw (\n)--(\a);
  \end{tikzpicture}
  \hfill
  \begin{tikzpicture}[scale=0.6] 
    \foreach \a in {1,2,...,7} \node[w] (\a) at (60*\a:1) {};
    \foreach[count=\n] \a in {2,3,...,7} \draw (\n)--(\a);
    \node[b] (c) at (0,0) {};
    \foreach \a in {2,4,6} \draw (c)--(\a);
  \end{tikzpicture}
  \hfill
  \begin{tikzpicture}[scale=0.5] 
    \node[w] (s) at (1,2) {};
    \foreach[count=\y] \n in {a,b,c} \node[w] (\n) at (2,\y) {};
    \foreach[count=\y] \n in {d,e,f} \node[w] (\n) at (3,\y) {};
    \node[w] (t) at (4,2) {};
    \draw (a)--(s)--(b)--(d)--(a)--(e)--(c)--(f)--(b)
          (s)--(c)  (f)--(t)--(e)  (t)--(d);
  \end{tikzpicture}
  \hfill
  \begin{tikzpicture}[scale=0.7] 
    \foreach[count=\x] \c in {w,b,w} {
      \node[\c] (u\x) at (\x,1) {};
      \node[\c] (l\x) at (\x,0) {};
      \draw (l\x)--(u\x);
    };
    \foreach \y in {u,l} \draw (\y1)--(\y2)--(\y3);
  \end{tikzpicture}
  \hspace*{\fill}

  \quad

  \hspace*{\fill}
  \begin{tikzpicture}[scale=0.5] 
    \node[b] (1) at (1,0) {};  \node[b] (2) at (2,0) {};
    \draw (1)--(2);
  \end{tikzpicture}
  \hfill
  \begin{tikzpicture}[scale=0.5] 
    \node[w] (1) at (1,0) {};  \node[b] (2) at (2,0) {};
    \node[b] (3) at (3,0) {};  \node[w] (4) at (4,0) {};
    \draw (1)--(2)--(3)--(4);
  \end{tikzpicture}
  \hfill
  \begin{tikzpicture}[scale=0.5] 
    \node[w] (1) at (1,0) {};  \node[w] (2) at (2,0) {};
    \node[b] (3) at (3,0) {};  \node[w] (4) at (4,0) {};
    \node[w] (5) at (5,0) {};
    \draw (1)--(2)--(3)--(4)--(5);
  \end{tikzpicture}
  \hspace*{\fill}

  \caption{The path $P_9$, the cycle $C_8$, the augmented bipartite wheel $BW^*_3$,
    the cube $Q_3$, an augmented domino, followed by the augmented paths $P_2^*$,
    $P_4^*$ and $P_5^*$.
  }
  \label{fig:cbff}
\end{figure}
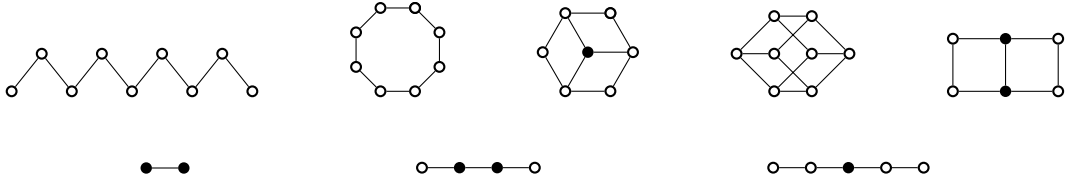

\begin{lemma} \label{l:bip fork-free}
  A bipartite graph is fork-free if and only if every connected component is a
  path, a cycle of even length, a $BW^*_3$, a cube $Q_3$, or can be obtained
  from a complete bipartite graph by removing at most two edges that form a
  matching, see Fig.~\ref{fig:cbff}.
\end{lemma}

Note that a graph is a $P_2^*$ if and only if it is complete bipartite, and
a $P_4^*$ if it can be obtained from a complete bipartite graph by removing
one edge. If we remove two independent edges we obtain the graphs represented
by the augmented domino in Fig.~\ref{fig:cbff}. $P_4^*$ and $P_5^*$ are induced subgraphs
of the augmented dominoes, and $P_2^*$ is in the same way contained in $P_4^*$. Clearly
every path is an induced subgraph of a suitable even cycle. Finally, $C_6$ is
a $K_{3,3}$ minus a perfect matching, and $Q_3$ is a $K_{4,4}$ minus a perfect
matching.

\begin{proof}
  It is easy to see that none of the connected bipartite graphs
  depicted in Fig.~\ref{fig:cbff} contains a fork as induced subgraph.
  (Adding false twins cannot produce a fork where no $P_4$ ends.)

  To prove the other implication we consider a connected bipartite graph
  $H=(V,E)$ that does not contain a fork. First we suppose that $H$ is a tree.
  If $H$ contains no vertex of degree three or more then $H$ is a path, and we
  are done. Otherwise, let $v$ be a vertex in $H$ of degree at least three. If
  any vertex in $H$ has distance at least two from $v$ then we have an induced
  fork since $H$ is a acyclic. Otherwise, every vertex is distance at most one
  from $v$ and $H$ is a star $K_{1,b}$ for some integer $b\ge3$. All these
  stars are complete bipartite graphs.

  Now suppose that $H$ contains a cycle, and let $C=(v_1,v_2,\dots,v_{2\ell})$
  (as $H$ is bipartite, $C$ must have even length) be a longest induced cycle
  in $H$. If $C=H$ then we are done. Otherwise, there is a vertex $v$ in $C$
  with degree at least three.

  Assume that there exists a vertex $w$ in $H$ in distance two from $C$. That is, there is a
  path $(w,u,v)$ where, without loss of generality, $v=v_2$. Now
  $\{v_1,v_2,v_3,u,w\}$ induces a fork in $G$ since $w$ has no neighbour in
  $C$, which is a contradiction. Hence every vertex of $H$ that does not belong to $C$ has
  a neighbour on $C$, and the same is true for every other cycle in $H$.
  The phrase ``every cycle is dominating" will refer to this property.

  We distinguish cases depending on the length of the longest cycle $C$ in $H$.
  First assume that $\ell \ge 3$; that is, $C$ has length at least six. We consider
  a vertex $u$ that does not lie on $C$. Let $v_2$ be a neighbour of $u$.
  If $\{v_{2i} \, :\,  i \in [\ell]\} \sm \Nb(u) \ne \es$ then we may assume
  that
  $u \in \Nb(v_2) \sm \Nb(v_4)$, which would cause a fork in $G$ induced by $v_1$,
  $v_2$, $v_3$, $v_4$ and $u$.
  Hence we have $\{v_{2i} \, :\, i \in [\ell]\} \subseteq \Nb(u)$. If $\ell \ge 4$
  then $v_2$, $v_4$, $u$, $v_6$ and $v_7$ induce a fork. Therefore we have
  $\ell=3$. We have $\{v_{2i-1} \, :\, i \in [\ell]\} \subseteq \Nb(u)$ or
  $\{v_{2i} \, :\, i \in [\ell]\} \subseteq \Nb(u)$ for every vertex $u$ of $H$
  that does not belong to $C$. Hence $H$ is a $BW^*_3$ or a cube $Q_3$,
  because adding a false twin to $Q_3$ would cause a fork. To see this,
note that $Q_3$ contains $C_6$, by removing two opposite vertices.
Then adding a false twin to any vertex in $C_6$ produces a fork.

  In the remaining case, every induced cycle of $H$ has length four. If $H$ is
  complete bipartite we are done. Otherwise we choose a maximal complete
  bipartite subgraph of $H$. More precisely, let $X$ and $Y$ be independent
  sets of $H$ such that
  \begin{enumerate}
  \item[(a)] \label{XY1} every vertex in $X$ is adjacent to every vertex in $Y$
        (that is, $X$ and $Y$ induce a complete bipartite subgraph in $H$),
  \item[(b)] \label{XY2} $|X| \ge 2$ and $|Y| \ge 2$
        (this is possible because $H$ contains a $4$-cycle),
  \item[(c)] \label{XY3} with respect to (a) and (b) the set 
        $X \cup Y$ has maximum size.
  \end{enumerate}
  Let $W = \Nb(X) \sm Y$ and $Z = \Nb(Y) \sm X$. Since every cycle in $H$ is
  dominating, $(W,X,Y,Z)$ is a partition of $V$. We split $Y$ into those
  vertices that have a non-neighbour in $Z$ and those
  that are adjacent to all vertices in $Z$ by setting
  $Y'  = Y \sm  \bigcap_{z \in Z} \Nb(z)$ and
  $Y'' = Y \cap \bigcap_{z \in Z} \Nb(z)$. Similarly
  $X'  = X \sm  \bigcap_{w \in W} \Nb(w)$ and
  $X'' = X \cap \bigcap_{w \in W} \Nb(w)$.
  Every vertex $z \in Z$ has a non-neighbour in $Y'$, otherwise $z$ would
  belong to $X$. If $z$ has two non-neighbours $y_1, y_2 \in Y'$ then, for
  every pair of vertices $x_1, x_2 \in X$, the $4$-cycle $(x_1,y_1,x_2,y_2)$
  would not dominate $z$. Therefore every vertex $z \in Z$ has exactly one
  non-neighbour in $Y'$. If $Z$ contains three vertices $z_1$, $z_2$ and
  $z_3$ with different non-neighbours $y_1, y_2, y_3 \in Y'$ then $(y_1, z_2,
  y_3, z_1, y_2, z_3)$ is a chordless $6$-cycle in $H$.
  But this contradicts the assumption of
  this case (namely, that all chordless cycles have length~$4$).

  For every vertex $z \in Z$ with neighbour $y_1 \in Y$ and non-neighbour
  $y_2 \in Y$ and every vertex $w \in W$ with neighbour $x \in X$, the
  vertices $w$, $y_2$, $x$, $y_1$ and $z$ induce a fork, unless $w$ and $z$
  are adjacent.
  Hence every vertex $w \in W$ is adjacent to every vertex $z \in Z$.
  Consequently the graph $H$ is `almost complete bipartite' with bipartition
  $W \cup Y$ and $X \cup Z$. The missing edges are between $W$ and $X$ or
  between $Z$ and $Y$. These non-edges form a matching, and therefore
  there are at most two of them, because the endpoints of three independent
  non-edges would induce a $C_6$.
\end{proof}

\begin{lemma} \label{l:bpw fork-free}
  For all integers $d \ge 1$ the fork-free graphs without
  induced $K_{d+1,d+1}$ have bipartite pathwidth at most $\max(4,d+2)$.
\end{lemma}

\begin{proof}
  The (bipartite) pathwidth of a disconnected graph is the maximum (bipartite)
  pathwidth of its connected components. Therefore we just need to check all
  the possibilities for connected induced bipartite subgraphs as listed in
  Lemma~\ref{l:bip fork-free}. For $n \ge 2$ the path $P_n$ has pathwidth $1$,
  $\pw(P_1)=0$, and for $n \ge 3$ the cycle $C_n$ has pathwidth $2$, see
  \eqref{eq:pw}.
  The other graphs from the list we embed into suitable complete
  bipartite graphs to bound their pathwidth using Lemma~\ref{l:pw is monotone}.

  We have $BW^*_3 \subseteq K_{3,b+3}$ where $b\ge1$ is the number of central
  vertices. This implies that $\pw(BW^*_3) \le \pw(K_{3,b+3}) = 3$. Similarly,
  $Q_3 \subseteq K_{4,4}$ and therefore $\pw(Q_3) \le \pw(K_{4,4}) = 4$.

  For each augmented domino $D$ there exists positive integers $a$ and $b$ such that
  $K_{a,b} \subseteq D \subseteq K_{a+2,b+2}$. Since $D$ does not contain
  $K_{d+1,d+1}$ we have $\min\{a,b\} \le d$, hence $\pw(D) \le
  \pw(K_{a+2,b+2}) \le d+2$. All other graphs from Lemma~\ref{l:bip fork-free}
  are subgraphs of augmented dominoes.
  Thus we see that every possible connected bipartite induced subgraph of a
  fork-free graph without induced $K_{d+1,d+1}$ has pathwidth at most
  $\max(4,d+2)$.
\end{proof}

\subsection{
 Graphs free of armchairs, stirrers and tripods} 

Let a \emph{hole} in a graph be a chordless cycle of length five or more.  The
induced subgraphs depicted in Fig.~\ref{fig:tas} are called \emph{armchair},
\emph{stirrer} and \emph{tripod}. A \emph{\fast graph} is a graph that
contains none of these three as an induced subgraph.
(Here ``\fast'' stands for ``free of armchairs,
stirrers and tripods''.)
This extends the class of \emph{monotone graphs}~\cite{DJM},
which also excludes all holes.

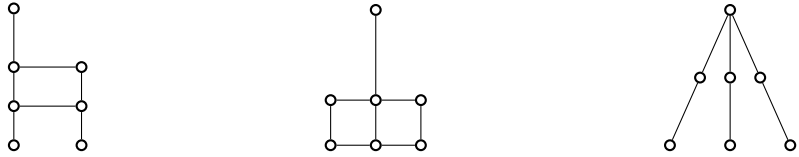
\begin{figure}[htbp]
  \hspace*{\fill}
  \begin{tikzpicture}[xscale=0.3,yscale=0.26]             
    \foreach \y in {0,4} \node[w] (0!\y) at (0,\y) {};
    \foreach \y in {2,7} \node[w] (0!\y) at (0,\y) {};
    \foreach \y in {0,4} \node[w] (3!\y) at (3,\y) {};
    \node[w] (3!2) at (3,2) {};
    \draw (0!0)--(0!2)--(0!4)--(0!7)  (3!0)--(3!2)--(3!4)--(0!4)  (0!2)--(3!2);
  \end{tikzpicture}
  \hspace*{\fill}
  \begin{tikzpicture}[scale=0.3]                          
    \node[w] (hh) at (3,6) {};
    \foreach \x in {1,5} {
      \draw (\x,2) node[w] (\x!2) {} (\x,0) node[w] (\x!0) {};
    };
    \draw (3,0) node[w] (3!0) {} (3,2) node[w] (3!2) {} ;
    \draw (hh)--(3!2)--(3!0)--(1!0)--(1!2)--(3!2)--(5!2)--(5!0)--(3!0);
  \end{tikzpicture}
  \hspace*{\fill}
  \begin{tikzpicture}[xscale=0.4,yscale=0.3]              
    \node[w] (hh) at (3,6) {};                            
    \node[w] (l1) at (2,3) {}; \node[w] (l2) at (1,0) {}; 
    \node[w] (m1) at (3,3) {}; \node[w] (m2) at (3,0) {}; 
    \node[w] (r1) at (4,3) {}; \node[w] (r2) at (5,0) {}; 
    \foreach \leg in {l,m,r} \draw (hh)--(\leg1)--(\leg2);
  \end{tikzpicture}
  \hspace*{\fill}
  \caption{The armchair, the stirrer and the tripod.}
  \label{fig:tas}
\end{figure}

\begin{lemma} \label{l:nmah} 
  For every vertex $w$ of a connected bipartite \fast graph $G$, the graph
  $G \sm \Nb(w)$ is hole-free.
\end{lemma}

\begin{proof}
  For a contradiction, assume that there exists a vertex $w$ such that
  $G \sm \Nb(w)$ contains a hole $C$. Since $G$ is bipartite we have
  $C=(u_1,u_2,\dots,u_{2\ell})$ with $\ell\geq 3$. As $G$ is connected there
  is a path from $w$ to $C$, which must contain at least~2 edges.
  Let $P = (w,\ldots, w',v,u_i)$ be a shortest path from $w$ to $C$, where
  $i\in [2\ell]$.  Observe that $w'$ has no neighbours on $C$ (or a shorter
  path would exist from $w$ to $C$.)

  First suppose that (at least one of) $u_{i-2}$ or $u_{i+2}$ are also
  adjacent to $v$. We claim that $\{u_{i-3},\dots, u_{i+1},v,w'\}$ and
  $\{u_{i-1},\dots,u_{i+3},v,w'\}$, respectively, induce an armchair in $G$.
  This follows since $G$ is bipartite and $w'$ has no neighbours on $C$.
  Otherwise, neither $u_{i-2}$ nor $u_{i+2}$ is adjacent
  to $v$ and $\{u_{i-2},\dots,u_{i+2},v,w'\}$ induces a tripod. Since both
  subgraphs are forbidden in \fast graphs, such a hole $C$ does not exist.
\end{proof}

\subsubsection{Maximum degree bound}

\begin{lemma} \label{l:pw monotone}
  For every bipartite hole-free \fast graph $G=(A,B,E)$ we have \\
  \[ \pw(G) \le \min\{\max\{d(v) \, :\,  v \in A\},\,  \max\{d(v) \, :\,  v \in B\}\}.\]
\end{lemma}

\begin{proof}
  A bipartite hole-free \fast graph is monotone. We rename the vertices in $A$
  by $1, 2, \dots, a$ and those in $B$ by $1', 2', \dots, b'$ such that the
  bi-adjacency matrix of $G$ with rows and columns in this order has the
  characteristic form of a staircase. Both $(\Nb[i])_{i=1}^{a}$ and
  $(\Nb[j'])_{j=1}^{b}$ are path decompositions of $G$ of width
  $\max\{d(v) \,:\,  v \in A\}$ and $\max\{d(v) \, :\,  v \in B\}$, respectively.
\end{proof}

\begin{lemma}
  A bipartite \fast graph $G$ has pathwidth at most $2\Delta(G)$.
\end{lemma}

\begin{proof}
  If $G$ is disconnected then its pathwidth is the maximum pathwidth of
  its connected components.  So we may assume that $G$ is connected, and choose any
vertex $w$.  Now $G \sm \Nb(w)$ is hole free, by Lemma~\ref{l:nmah},
and so $G \sm \Nb(w)$ has a path decomposition $(B_i)_{i=1}^{t}$
  of width at most $\Delta(G)$, by Lemma~\ref{l:pw monotone}.
Consequently $(B_i \cup \Nb(w))_{i=1}^{t}$
  is a path decomposition of $G$ and its width is at most $2\Delta(G)$.
\end{proof}

\begin{corollary}
  For every positive integer $d$, a bipartite graph that does not contain
  a tripod, an armchair, a stirrer or a star $K_{1,d+1}$ as induced subgraph
  has pathwidth at most $2d$.
\end{corollary}

This implies that \fast graphs with degree bound $d$ have bipartite pathwidth at most $2d$.
However, we will improve on this below.

\subsubsection{Bound on the size of complete bipartite subgraphs}

For a positive integer $n$ let $[n] = \{1,2,\dots,n\}$ and for $n \in [m]$
let $[n,m] = \{n,n+1,\dots,m\}$. We consider a monotone graph $G=(L,R,E)$
where $L = [\ell]$ and $R = [r]'$, where the latter is $\{j' \, :\,  j \in [r]\}$.
For $d \ge 0$ let $L_i^d = [i,i+d-1]$ and $R_j^d = [j,j+d-1]'$. If $d$ is fixed
then we abbreviate $x+d-1$ by $\hat{x}$.

We define $\psi(G) = \max \{d \, :\, K_{d,d} \subseteq G\}$. Since
$\pw(K_{d,d}) = d$ holds for all $d \ge 1$, it follows from Lemma~\ref{l:pw is monotone}
that $\psi(G) \le \pw(G)$
for all graphs $G$ with at least one edge.

\begin{lemma} \label{l:d<2b}
  For a bipartite \fast graph we have $\delta(G) \le 2\,\psi(G)$.
\end{lemma}

\begin{proof}
Let $k=\psi(G)$, and suppose, for a contradiction, that $\delta(G) \ge 2k+1$.
Let $i\in L$ be any vertex,
and $j'$ be such that $\{i,t'\}\in E$ for all $t\in[j-k,j+k]$. Such an $j$ exists
as $G$ is monotone and bipartite and
$\deg(i)\geq 2k+1$. Now let $u=\max\{s\, :\, \{i-s,j'\}\in E\}$.
Clearly $u\geq 0$, and $\{t,j'\}\in E$ for all $t\in[i-u,i-u+2k]$. This interval for $t$
is guaranteed by $\deg(i)\geq 2k+1$. Now, by the monotone property of \fast graphs, we have
$\{i-u,j'-k\}\in E$ and $\{i-u+2k,j'+k\}\in E$. Hence, again by the monotone property,
$\{t,r'\}\in E$ for all $t\in[i-u,i]$, $r\in[j'-k,j]$, and for all $t\in[i,i+2k-u]$, $r\in[j',j'+k]$.
The former is a $(u+1)\times(k+1)$ biclique, and the latter is a $(2k-u+1)\times(k+1)$ biclique, as illustrated in Figure~\ref{fig:d<2b}.
If $u\geq k$ then the former contains a $(k+1)\times (k+1)$ biclique,
and, if $u\leq k$ then the latter contains a $(k+1)\times (k+1)$ biclique,
contradicting $\psi(G)=k$.

 \begin{figure}[htbp]
\begin{center}
    \hspace*{\fill}
    \begin{tikzpicture}[scale=0.4,font=\small]
    \draw (2.75,-1) rectangle (14,9) ;
      \foreach \n/\y in {i-u/8, i/6, i-u+2k/2} \node[left] (\n) at (2.0,\y) {$\n$};
      \foreach \n/\x in {(j-k)/4, j/8,(j+k)/12} \node[i] (\n) at (\x,10) {$\n'$};
      \draw[-,thin] (4,-1)--(4,9) (8,-1)--(8,9) (12,-1)--(12,9) ;
      \draw[-] (2.75,2)--(14,2) (2.75,6)--(14,6) (2.75,8)--(14,8) ;
      \draw[fill=lightgray!60!white] (4.0,6.0) rectangle (8.0,8.0);                         ;
      \draw[fill=lightgray!60!white] (8.0,2.0) rectangle (12.0,6.0);
\node at  (17,0) {\phantom{hi}};
  \end{tikzpicture}
    \hspace*{\fill}
    \caption{Two bicliques which must be present if $\delta(G)>2\psi(G).$}
    \label{fig:d<2b}
\end{center}
  \end{figure}
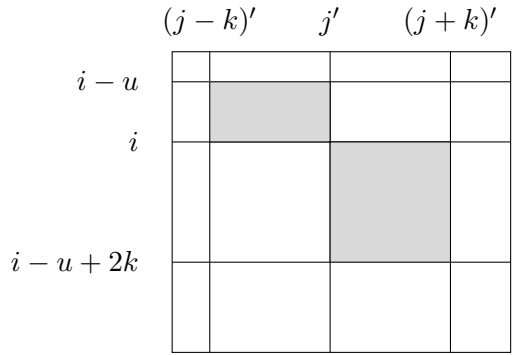
\end{proof}
The bound in Lemma~\ref{l:d<2b} is tight, by the following construction. Let $G$ have $L=[n]$, $R=[n]'$  and
$E=\{(i,j') \, :\, i\in[n], j'\in[i',(i+\delta \mod n)']\}$. See Figure~\ref{fig:C6}. It is easy to see that this graph has minimum degree $\delta$, and the largest $k\times k$ biclique has $k\leq\delta/2$, so $\delta\geq2\psi$.
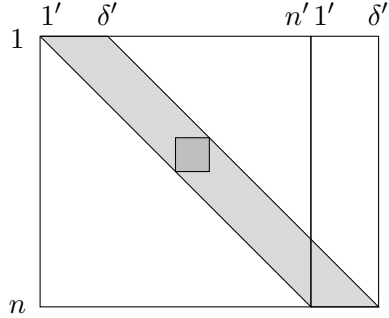
\begin{figure}[htbp]
  \centering
  \begin{tikzpicture}[scale=0.3,font=\small]
  \draw[fill=lightgray!60!white] (0,12)--(12,0)--(15,0)--(3,12)--(0,12);
  \draw (0,0) rectangle (12,12) ;
  \draw (12,0) rectangle (15,12) ;
  \draw[fill=lightgray] (6,6) rectangle (7.5,7.5) ;
    \draw (-1,0) node {$n$} (-1,11.9) node {$1$};
  \draw (0.5,13) node {$1'$} (3,13) node {$\delta'$};
  \draw (11.4,13) node {$n'$} (12.6,13) node {$1'$} (15,13) node {$\delta'$};
    \end{tikzpicture}
  \caption{The bound in Lemma \ref{l:d<2b} is tight.}
  \label{fig:C6}
\end{figure}


\begin{thm}\label{thm:mono}
  For every monotone graph $G$ with at least one edge
  $\pw(G) \le 2\,\psi(G) - 1$ holds.
\end{thm}

\begin{proof}
  If $G$ is disconnected then its pathwidth is the maximum pathwidth of its
  connected components. Therefore we may assume that $G$ is connected.
If $G$ has no edges then $G$ is an isolated
vertex, and we can take a path decomposition consisting of one bag containing
this vertex.  From now on we assume that $G$ has at least one edge.

Denote the
  partite sets of $G$ by $L=[1,n]$ and  $R=[1,m]'$, and let $d = \psi(G)$. As before
  let $\hat{j} = j+d-1$.
  We assume $L$ and $R$ are numbered such that the bipartite adjacency
  matrix $A$ of $G$ does not contain the following submatrices:
  \[ \left(\begin{array}{cc} 1 & 1 \\ 1 & 0 \end{array} \right) \qquad
     \left(\begin{array}{cc} 0 & 1 \\ 1 & 0 \end{array} \right) \qquad
     \left(\begin{array}{cc} 0 & 1 \\ 1 & 1 \end{array} \right) \]
  We construct a path decomposition $(B_i)_{i=1}^{t}$ of $G$ where
  $t=n+m-2d+1$.
  Each bag $B_i$ is of the form $[\ell_i,\hat{\ell}_i] \cup  [r_i,\hat{r}_i]'$.
  For the first bag we have $\ell_1=1$ and $r_1=1$, and the last bag has
  $\hat{\ell}_t=n$ and $\hat{r}_t=m$. For two consecutive bags $B_i$ and
  $B_{i+1}$ we have either
\begin{equation}
\label{taxi-driver}
 \ell_{i+1} = \ell_i \,\, \text{ and  }\,\, r_{i+1} = r_i + 1,\,\,
\text{ or } \ell_{i+1} = \ell_i+1 \,\, \text{ and }\,\, r_{i+1} = r_i.
\end{equation}
That is, we move a window
  of size $d \times d$ over $A$ from the top-left position to the bottom-right
  one. In each step we move it either one unit to the right or one unit down.
  To obtain a path decomposition we have to do this in such a way that every
  entry $1$ in $A$ is covered by the window at least once. This way we ensure
  that for every edge $e$ of $G$ there is an index $i \in [1,t]$ such that
  $e \subseteq B_i$.

  Let $B_i = [\ell_i,\hat{\ell}_i] \cup [r_i, \hat{r}_i]$.
  If  $\ell_i \in L$ is adjacent to $(r_i+d)'$ then we move the window to the right,
  that is,  $\ell_{i+1} = \ell_i$ and $r_{i+1} = r_i + 1$.
  If $r_i' \in R$ is adjacent to $\ell_i+d$ then we move the window down,
  that is, $\ell_{i+1} = \ell_i+1$ and $r_{i+1} = r_i$.  If both of these conditions hold then
  $[\ell_i,\ell_i+d] \cup [r_i, r_i+d]'$ induces a $K_{d+1,d+1}$
  in $G$, contradicting $\psi(G)=d$. If neither $\{\ell_i,(r_i+d)'\}$ nor
  $\{\ell_i+d, r_i'\}$ is an edge then it does not matter whether we move the
  window down or right, as long as $\hat{\ell}_{i+1} \le n$ and
  $\hat{r}_{i+1} \le m$.  (When both directions are possible we choose one arbitrarily.)

  We now check that $(B_i)_{i=1}^t$ is indeed a path decomposition of $G$.
  Condition \ref{pw1} from the definition is fulfilled because we start
  at the top-left corner of the adjacency matrix ($\ell_1=1$ and $r_1=1$), stop
  at the bottom-right corner ($\hat{\ell}_t=n$ and $\hat{r}_t=m$), and using
 (\ref{taxi-driver}).
  Since $G$ is connected we know that $\{ 1, 1'\}, \{ n, m'\}$ are both edges of $G$,
and by construction and definition of $\psi$, the window sweeps over every entry of
$A$ which equals~1.  (For example, if the current window is at $(\ell_i,r_i')$ then the
window will not move down if $\{ \ell_i,(r_i+1)'\}$ is an edge: rather, the window will
move to the right and cover this entry.)
This shows that Condition~\ref{pw2} holds.
Condition~\ref{pw3} is satisfied because
  we move the window only to the right or down, never to the left or up.
  Finally, all the bags $B_i$ have size $2d$. Therefore $\pw(G) \le 2d-1$.
\end{proof}

Together with Lemmas \ref{l:nmah} and \ref{l:d<2b} we obtain the following
corollary.

\begin{thm}\label{thm:tasf}
  For every integer $d \ge 1$, a bipartite graph that does not contain
  an armchair, a stirrer, a tripod or a $K_{d+1,d+1}$ as an induced subgraph
  has pathwidth at most $4d-1$.
\end{thm}

\begin{proof}
  We have $\psi\leq d$. Without loss of generality, we assume that $G$ is
  connected.  Let $v$ be a vertex of minimum degree in $G$, so $v$ has degree
  $\delta =\delta(G)$.  By Lemma~\ref{l:nmah}, the graph $G \sm \Nb(v)$ is
  hole-free, and hence $\pw(G \sm \Nb(v))\leq 2\psi- 1$, by
  Theorem~\ref{thm:mono}.  Then applying the second statement of
  Lemma~\ref{l:pw is monotone} shows that $\pw(G)\leq 2\psi- 1+\delta\leq
  4\psi-1$, using the fact that $|\Nb(v)| = \delta\leq 2\psi$, by
  Lemma~\ref{l:d<2b}.
\end{proof}

The bound of Theorem~\ref{thm:tasf} is almost tight.
Let $G$ be the bipartite \fast graph depicted in Fig.~\ref{fig:C6}.
We claim that the pathwidth of $G$ is $4d-2$.
Let $S$ be a set which contains the intersection of the neighbourhoods of two successive
vertices $i,i+1$, or  $j',(j+1)'$. Thus $|S|\geq d-1$.
Then, by Lemma~\ref{l:pw is monotone},
\[ \pw(G) \leq \pw(G \sm S) + |S|\leq d -1 + 2\psi-1\leq 4d-2,\]
since the graph of Fig.~\ref{fig:C6} is tight for Lemma~\ref{l:d<2b}.

\section{Conclusions and further work}

It is clearly NP-hard in general to determine the bipartite pathwidth of a graph,
since it is NP-complete to determine the pathwidth of a bipartite graph. However,
we need only determine whether $\bpw(G)\leq d$ for some constant $d$. The complexity of
this question is  not currently known, though as mentioned earlier, Mann and Mathieson~\cite{MM}
report that the problem is W[1]-hard in the worst case.
 Bodlaender~\cite{HLB:pw} has shown that the question $\pw(G)\leq d$,
can be answered in $O(2^{d^2}n)$ time. However, this implies nothing about $\bpw(G)$, since we have
seen that $\bpw$ may be bounded for graph classes in which $\pw$ is unbounded.

We have therefore examined some classes of graphs where we can guarantee that $\bpw(G)\leq d$,
for some known constant $d$. Here our recognition algorithm is simply detection of excluded
induced subgraphs,
and we leave open the possibility of more efficient recognition.

In the case of claw-free graphs we have obtained stronger sampling results using log-concavity.
This raises the question of how far log-concavity extends in this setting.
For example, does it hold for fork-free graphs? More ambitiously, does some generalisation
of log-concavity hold for graphs of bounded bipartite pathwidth?

Where log-concavity holds, it allows us to approximate the number of independent sets of a given size.
However, there is still the requirement of ``amenability''~\cite{JS}. Jerrum, Sinclair and Vigoda~\cite{JSV}
have shown that this can be dispensed with in the case of matchings in bipartite graphs. A natural question is:
how far does this extend to claw-free graphs?
In~\cite{DJMV}, it is shown that the results of~\cite{JSV} carry over to the class (fork,\,odd hole)-free graphs, which strictly contains the class of line graphs of bipartite graphs.

An extension would be to consider \emph{bipartite treewidth}, $\btw(G)$. Since $\tw(G)= O(\pw(G)\log n)$ \cite[Thm.~66]{HLB:arb},
our results here immediately imply that bounded bipartite treewidth implies \emph{quasipolynomial} mixing time for the Glauber dynamics. Can this be improved to polynomial time, or can some other approach give this?

Finally, can other approaches to approximate counting be employed for the independent set problem in
these graph classes? As mentioned earlier, Patel and Regts~\cite{PR} used the 
Taylor expansion approach for claw-free graphs, and it follows from
a result of Bencs~\cite{bencs} that the Taylor expansion approach can be applied to 
bounded-degree fork-free graphs, giving a FPTAS for approximating $P_\lambda(G)$ in
these cases.  Can these results be extended to larger classes of 
graphs?

\subsubsection*{Acknowledgements} 
We would like to thank the anonymous referee for their helpful comments and for
bringing~\cite{alg,clv} to our attention. We also thank Ferenc Bencs for letting
us know that the results of~\cite{bencs} lead to an FPTAS
for the independence polynomial of bounded-degree fork-free graphs.


\begin{thebibliography}{99}

\bibitem{Aldous} D.~Aldous,
  On the Markov chain simulation method for uniform combinatorial distributions
  and simulated annealing,
  \textsl{Probability in the Engineering and Informational Sciences} \textbf{1} (1987), 33--46.

\bibitem{Alekseev} V.\,E.~Alekseev,
  Polynomial algorithm for finding the largest independent sets in graphs without forks,
  \textsl{Discrete Applied Mathematics} \textbf{135} (2004), 3--16.

\bibitem{alg}
N.~Anari, K.~Liu and S.\,O.~Gharan,
Spectral independence in high-dimensional expanders and applications to the
hardcore model, 
in \textsl{Proc.\ 61st Annual Symposium on Foundations of Computer Science} (FOCS 2020),
IEEE, 2020. 

\bibitem{barvinok-book}
A.~Barvinok, \emph{Combinatorics and Complexity of Partition Functions},
Springer, Cham, 2016.

\bibitem{barvinok} A.~Barvinok,
  Computing the partition function of a polynomial on the Boolean cube,
  in: M.~Loebl, J.~Ne\v{s}et\v{r}il and R.~Thomas (eds)
  \textsl{A Journey Through Discrete Mathematics}, Springer, 2017, pp.~135--164.

\bibitem{BGKNP}
  M.~Bayati, D.~Gamarnik, D.~Katz, C.~Nair and P.~Tetali,
  Simple deterministic approximation algorithms for counting matchings,
  \textsl{Proc.\ 39th ACM Symposium on Theory of Computing} (STOC 2007), ACM, 122--127.

\bibitem{beineke} L.~Beineke,
  Characterizations of derived graphs,
  \textsl{Journal of Combinatorial Theory} \textbf{9} (1970), 129--135.

\bibitem{bencs}
F.~Bencs, Christoffel--Darboux type identities for the independence polynomial,
\emph{Combinatorics, Probability and Computing} {\bf 27(5)} (2018), 716--724.

\bibitem{HLB:pw} H.\,L.~Bodlaender,
  A linear time algorithm for finding tree-decompositions of small treewidth,
  \textsl{SIAM Journal on Computing} \textbf{25} (1996), 1305--1317.

\bibitem{HLB:arb} H.\,L.~Bodlander,
  A partial $k$-arboretum of graphs with bounded treewidth,
  \textsl{Theoretical Computer Science} \textbf{209} (1998), 1--45.

\bibitem{BKW} P.~Bonsma, M.~Kami{\' n}ski and M.~Wrochna,
  Reconfiguring independent sets in claw-free graphs,
  in: 
  \textsl{Algorithm Theory -- SWAT 2014}.
  Springer LNCS \textbf{8503} 2014, 86--97.

\bibitem{Branden} P.~Br\"and\'en,
  Unimodality, log-concavity, real–rootedness and beyond,
  Chapter~7 in \textsl{Handbook of Enumerative Combinatorics}, Chapman and Hall/CRC, 2015.

\bibitem{survey} A.~Brandst\"adt, V.\,B.~Le, and J.~Spinrad,
  \textsl{Graph Classes: A Survey},
  SIAM Monographs on Discrete Mathematics and Application,
  Philadelphia, 1999.

\bibitem{clv}
Z.~Chen, K.~Liu and E.~Vigoda,
Rapid mixing of Glauber dynamics up to uniqueness via contraction,
in \textsl{Proc.\ 61st Annual Symposium on Foundations of Computer Science} (FOCS 2020),
IEEE, 2020. 

\bibitem{CS} M.~Chudnovsky and P.~Seymour,
  The roots of the independence polynomial of a clawfree graph,
  \textsl{Journal of Combinatorial Theory} (\textsl{Series B}) \textbf{97} (2007), 350--357.

\bibitem{Curt} R.~Curticapean,
  Counting matchings of size $k$ is \#W[1]-hard,
  in: \textsl{Automata, Languages, and Programming (ICALP 2013)}.
  Springer LNCS \textbf{7965}, 2013, 352--363.

\bibitem{DS93} P.~Diaconis and L.~Saloff-Coste,
  Comparison theorems for reversible {M}arkov chains.
  \textsl{Annals of Applied Probability}, 3:696--730, 1993.

\bibitem{DS91} P.~Diaconis and D.~Stroock.
  Geometric bounds for eigenvalues of {M}arkov chains.
  \textsl{Annals of Applied Probability}, 1:36--61, 1991.

\bibitem{diestel} R.~Diestel, \textsl{Graph Theory},
  4th Edition, Graduate Texts in Mathematics \textbf{173}, Springer, 2012.

\bibitem{DGGJ03} M.~Dyer, L.\,A.~Goldberg, C.~Greenhill and M.~Jerrum,
  On the relative complexity of approximate counting problems,
  \textsl{Algorithmica} \textbf{38} (2003) 471--500.

\bibitem{DG} M.~Dyer and C.~Greenhill,
  On Markov chains for independent sets,
  \textsl{Journal of Algorithms} \textbf{35} (2000), 17--49.

\bibitem{DM} M.~Dyer and H.~M{\"u}ller,
  Counting independent sets in cocomparability graphs,
  \textsl{Information Processing Letters} \textbf{144} (2019), 31--36.

\bibitem{DGM} M.~Dyer, C.~Greenhill and H.~M{\"u}ller,
  Counting independent sets in graphs with bounded bipartite pathwidth,
  in: \textsl{Graph-Theoretic Concepts in Computer Science (WG 2019)},
  Springer LNCS \textbf{11789}, 2020, 298--310.

 \bibitem{DJM} M.~Dyer, M.~Jerrum and H.~M{\"u}ller,
  On the switch Markov chain for perfect matchings,
 \textsl{Journal of the Association for Computing Machinery} \textbf{64} (2017), Art.~12.

\bibitem{DJMV} M. Dyer, M. Jerrum, H. M\"uller and K. Vu\v{s}kovi\'{c},
  Counting weighted independent sets beyond the permanent,
     Preprint, 2019. \arxiv{1909.03414}

\bibitem{edmonds} J.~Edmonds, Paths, trees, and flowers,
  \textsl{Canadian Journal of Mathematics} \textbf{17} (1965), 449--467.

\bibitem{EHSVY} C.~Efthymiou, T.~Hayes, D.~\v{S}tefankovi\v{c}, E.~Vigoda and Y.~Yin,
  Convergence of MCMC and loopy BP in the tree uniqueness region for the hard-core model,
  in \textsl{Proc.\ 57th Annual Symposium on Foundations of Computer Science} 
(FOCS 2016), IEEE, 2016, 704--713.

\bibitem{greenhill} C.~Greenhill,
  The complexity of counting colourings and independent sets in sparse graphs and hypergraphs,
  \textsl{Computational Complexity} \textbf{9} (2000), 52--72.

\bibitem{hamidoune} Y.\,O.~Hamidoune,
  On the numbers of independent $k$-sets in a claw free graph,
  \textsl{Journal of Combinatorial Theory} (\textsl{Series B}) \textbf{50} (1990), 241--244.

\bibitem{HaMa} G.~Harris and C.~Martin,
  The roots of a polynomial vary continuously as a function of the coefficients,
  \textsl{Proc.~of the AMS} \textbf{100} (1987), 390--392.

\bibitem{HSV} N.\,J.\,A.~Harvey, P.~Srivastava and J.~Vondr\'{a}k,
  Computing the independence polynomial: from the tree threshold down to the roots,
   in \textsl{Proc.\ of the 29th Annual ACM-SIAM Symposium on Discrete Algorithms}
   (SODA 2018), 1557--1576

\bibitem{jerrumbook} M.~Jerrum,
  \textsl{Counting, Sampling and Integrating: Algorithms and Complexity},
  Lectures in Mathematics -- ETH Z\"{u}rich, Birkh\"{a}user, Basel, 2003.

\bibitem{JS} M.~Jerrum and A.~Sinclair,
  Approximating the permanent,
  \textsl{SIAM Journal on Computing} \textbf{18} (1989), 1149--1178.

\bibitem{JSV} M.~Jerrum, A.~Sinclair and E.~Vigoda, A polynomial-time approximation algorithm
  for the permanent of a matrix with non-negative entries,
  \textsl{Journal of the ACM} \textbf{51} (2004), 671--697.

\bibitem{JVV} M.\,R.~Jerrum, L.\,G.~Valiant and V.\,V.~Vazirani,
  Random generation of combinatorial structures from a uniform distribution,
  \textsl{Theoretical Computer Science} \textbf{43} (1986), 169--188.

\bibitem{MM} R.~Mann and L.~Mathieson, personal communication.

\bibitem{JM} J.~Matthews,
  Markov chains for sampling matchings,
  PhD Thesis, University of Edinburgh, 2008.

\bibitem{LV} M.~Luby and E.~Vigoda, Approximately counting up to four,
  in \textsl{Proc.\ 29th Annual ACM Symposium on Theory of Computing} (STOC 1995),
  ACM, 1995, 150--159.

\bibitem{Minty} G.\,J.~Minty,
  On maximal independent sets of vertices in claw-free graphs,
  \textsl{Journal of Combinatorial Theory, Series B}, \textbf{28} (1980), 284--304,

\bibitem{PR} V.~Patel and G.~Regts,
  Deterministic polynomial-time approximation algorithms for partition
  functions and graph polynomials, \textsl{SIAM Journal on Computing} 46 (2017), 1893--1919.

\bibitem{PrBa83} J.\,S.~Provan and M.\,O.~Ball,
  The complexity of counting cuts and of computing the probability that a
  graph is connected,
  \textsl{SIAM Journal on Computing} \textbf{12} (1983) 777--788.

\bibitem{RS} N.~Robertson and P.\,D.~Seymour,
  Graph minors I: Excluding a forest,
  \textsl{Journal of Combinatorial Theory, Series B}, \textbf{35} (1983), 39--61.

\bibitem{sinclair} A.~Sinclair,
  Improved bounds for mixing rates of Markov chains and multicommodity flow,
  \textsl{Combinatorics, Probability and Computing} \textbf{1} (1992), 351--370.

\bibitem{sly} A.~Sly,
  Computational transition at the uniqueness threshold, in \textsl{Proc.\ 51st
  IEEE Symposium on Foundations of Computer Science} (FOCS 2010), IEEE, 2010, 287--296.

\bibitem{Vadhan} S.\,P.~Vadhan,
  The complexity of counting in sparse, regular, and planar graphs,
  \textsl{SIAM Journal on Computing} \textbf{31} (2001), 398--427.

\bibitem{weitz} D.~Weitz,
  Counting independent sets up to the tree threshold,
  in \textsl{Proc.\ 38th ACM Symposium on Theory of Computing} (STOC 2006), ACM, 2006, 140--149.

\end{thebibliography}
\end{document}